\documentclass[12pt]{article}

\usepackage[T1]{fontenc}
\usepackage{lmodern}
\usepackage{microtype}
\usepackage[margin=1.05in]{geometry}
\usepackage{amsmath,amsthm,amssymb,amsfonts,mathtools}
\usepackage{pdflscape}
\usepackage{booktabs}
\usepackage[inline]{enumitem}
\usepackage{xcolor}
\usepackage{caption}
\usepackage{hyperref}
\usepackage[nameinlink,capitalize,noabbrev]{cleveref}
\usepackage{natbib}

\hypersetup{
  colorlinks=true,
  linkcolor=blue!50!black,
  citecolor=blue!50!black,
  urlcolor=blue!50!black,
}

\renewcommand{\arraystretch}{1.15}
\setlist[itemize]{leftmargin=2em}
\setlist[enumerate]{leftmargin=2.2em}

\theoremstyle{plain}

\newtheorem{theorem}{Theorem}
\newtheorem{proposition}{Proposition}
\newtheorem{lemma}{Lemma}
\newtheorem{corollary}{Corollary}

\theoremstyle{definition}

\newtheorem{definition}{Definition}
\newtheorem{remark}{Remark}
\newtheorem{example}{Example}

\crefname{theorem}{Theorem}{Theorems}
\crefname{proposition}{Proposition}{Propositions}
\crefname{lemma}{Lemma}{Lemmas}
\crefname{corollary}{Corollary}{Corollaries}
\crefname{definition}{Definition}{Definitions}
\crefname{remark}{Remark}{Remarks}
\crefname{example}{Example}{Examples}

\Crefname{theorem}{Theorem}{Theorems}
\Crefname{proposition}{Proposition}{Propositions}
\Crefname{lemma}{Lemma}{Lemmas}
\Crefname{corollary}{Corollary}{Corollaries}
\Crefname{definition}{Definition}{Definitions}
\Crefname{remark}{Remark}{Remarks}
\Crefname{example}{Example}{Examples}

\numberwithin{equation}{section}

\DeclareMathOperator{\supp}{supp}
\DeclareMathOperator{\BR}{BR}
\DeclareMathOperator{\NE}{NE}
\DeclareMathOperator*{\argmax}{arg\,max}
\DeclareMathOperator*{\argmin}{arg\,min}

\newcommand{\R}{\mathbb{R}}

\begin{document}
\title{How damaging is zero-sum thinking to an agent's interests when the world is positive-sum?\footnote{The authors conducted this study without external funding and have no financial or non-financial conflicts of interest to disclose.}}
	\author{Shaun Hargreaves Heap\footnote{Department of Political Economy, King's College London, London, UK. E-mail: s.hargreavesheap@kcl.ac.uk.} \and Mehmet Mars Seven\footnote{Department of Political Economy, King's College London, London, UK. E-mail: mehmet.mars.seven@kcl.ac.uk.}}

	%\date{}
	\date{\today}
	\maketitle
    \begin{abstract}
We study whether zero-sum decision rules, maximin and minimax, harm agents' interests in positive-sum games relative to Nash equilibrium behaviour or, more generally, than best response behaviour. Contrary to an influential evolutionary view, we give illustrations where maximin serves an agent's interests better than Nash equilibrium behaviour. Two new sets of results show that these illustrations are not idiosyncratic. First, for any selected Nash equilibrium in cardinal games, we construct a strategically equivalent game where a maximin profile yields the same pay-offs and we fully characterise when maximin can be made to Pareto dominate that Nash equilibrium and the entire Nash equilibrium set. Second, we show that the relevant Maximin Pareto dominance classes are not knife-edge: they are generically interior, and strict maximin profile dominance occurs on a non-empty open set if and only if a player has at least two actions and the other has at least three. \textit{JEL}: C72, D81, D91
\end{abstract}
    
\section{Introduction}

There is evidence that zero-sum thinking is a distinct and prevalent cognitive style in rich countries (see Chinoy et al., \citeyear{chinoy2026}). It also appears to be more common among the young than the old, suggesting that such thinking may be increasing over time. Such a trend would also be consistent with the fact that zero-sum thinking is associated with support for Populist parties and their support has also been growing in these countries (see Burn-Murdoch, \citeyear{burnmurdoch2025}). In this paper, we are concerned with the possibility that zero-sum thinking may affect how individuals act on their preferences. In particular, suppose individuals believing the world to be zero-sum \underline{always} adopt a zero-sum decision rule (e.g.\ maximin). The question we address is: would the wholesale use of such a zero-sum rule harm their interests when they interact in positive-sum settings?

To the best of our knowledge, we are the first to consider this question systematically. For instance, in one recent discussion of zero-sum thinking, \citet{chinoy2026}, focus, using survey evidence, on the association between zero-sum thinking and particular policy preferences. They do not consider its possible effect on how people decide to act on their preferences: i.e. their decision rule.\footnote{For example, \citet{chinoy2026} find that `a more zero-sum mindset is strongly associated with more support for government redistribution, race- and gender-based affirmative action, and more restrictive immigration policies' (p.~1052).} Likewise, \citet{ali2025} focus on zero-sum thinking in elections but they do not consider how zero-sum thinking might affect decision rules. Instead, they make zero-sum thinking affect how people process information. In a theoretical model of elections, they show that asymmetric information and distributional uncertainty can sustain an equilibrium in which rational voters reject an ex ante beneficial policy because support for it by others is interpreted as bad news for themselves.\footnote{See also \citet{ali2021} who provide experimental evidence that subjects are more successful at reasoning about the behaviour of a better-informed player when interests are in conflict than when interests are aligned.}

Thus, our distinctive contribution to the current discussion of zero-sum thinking comes from considering how it might affect people's decision rule. This is potentially important because it would not be surprising if zero-sum thinking was to have an effect on decision rules. We know, for example, that maximin is a rational decision rule in zero-sum games, as shown by \citet{neumann1928} and \citet{neumann1944}. Hence, if individuals believe the world is zero-sum, then they might plausibly always use the maximin rule (i.e.\ in positive- as well as zero-sum games) rather than the Nash equilibrium strategies that would more normally follow from an expected utility maximisation decision rule in positive-sum games. Since the world presents both zero- and positive-sum settings, it therefore becomes interesting to know whether, in such cases, the wholesale use of a zero-sum rule like maximin damages a person's interests relative to the benchmark that comes from the selection of Nash equilibrium strategies.

We consider a variety of possible decision rules that a zero-sum thinker might use in positive-sum games: maximin; minimax; and relative-maximin. We focus for reasons that we explain in Section \ref{sec:zero-sum-rules} on maximin and minimax decision rules. We show that neither of these rules (and nor, incidentally, does relative-maximin) always cause damage to a person using them when compared with what would have happened had they played Nash equilibrium strategies. Indeed, there are some positive-sum settings where such zero-sum rules are positively advantageous. Thus, there appears to be no general pay-off inferiority reason for supposing that such zero-sum thinking will disappear within a population.

While our contribution to the discussion of zero-sum thinking comes from being, to the best of our knowledge, the first to explore the consequences of adopting maximin or minimax wholesale, there are, however, several strands in the wider literature that our analysis also addresses. We now mention these both to bring out our more general contribution and to help explain the structure of the paper.

First, there is a background presumption or received view, we suggest, that such zero-sum thinking damages an agent's interest. Adam Smith (\citeyear{smith1776}) is perhaps the most obvious original source for this view. He famously attacked Mercantilist zero-sum thinking because such thinking led to action (e.g.\ policies like tariffs) that actually harmed a country (see Book IV, Chapter II and III). Likewise, \citet{alchian1950} and \citet{friedman1953} supply an influential evolutionary argument in support of this view. To use Friedman's famous example, if a corporation does not adopt decision rules that are the equivalent of (`as if') profit maximising, then they will not survive in a competitive marketplace. Hence there would seem to be good evolutionary reasons for supposing that rational behaviour for an agent is maximising with respect to the satisfaction of whatever are their interests/objectives since any deviation can only potentially damage those interests and so undermine their survival prospects.

Against this background/received view, we show in Section \ref{sec:illustration} that there are symmetric 2-person games where the maximin decision rule does strictly better when playing against a fellow maximin player than two Nash equilibrium players do against each other (in natural notation (a) below); and does as well, when playing against the Nash strategy, as does the Nash player when playing against maximin (b).
\begin{equation*}
u_1(M,M) > u_1(N,N),
\tag{a}
\end{equation*}
\begin{equation*}
u_1(M,N) \geq u_2(M,N).
\tag{b}
\end{equation*}

In a population of maximin and Nash equilibrium rule followers, the expected returns to the use of each decision rule are, therefore, given by the following, where $\alpha > 0$ and $\beta > 0$ are the proportion of mutual interactions for, respectively, maximin and Nash rule followers.
\[
E(M) = \alpha u_1(M,M) + (1-\alpha)u_1(M,N),
\]
\[
E(N) = \beta u_1(N,N) + (1-\beta)u_2(M,N).
\]
When there is no difference in the frequencies of mutual rule encounters (i.e.\ $\alpha = \beta$), it follows that (a) and (b) ensure that $E(M)>E(N)$. In other words, when the only difference in the evaluation of the decision rules is the decision rule itself, maximin does better than Nash. We regard this as the cleanest comparison of the two rules because it brings out the difference that is due solely to the rule itself rather than some other difference between rule followers like the frequencies of their mutual interactions. Nevertheless, we also consider an extreme case where $\alpha \neq \beta$ by asking whether a population comprising of one strategy could be invaded by another. In this context, we also show (c) and (d).\footnote{We state condition (c) separately because it is the relevant ESS condition, although it follows from (a), (b), and (d).} (c) ensures that, in the reduced game with only maximin and Nash strategies, maximin is an evolutionarily stable strategy (ESS) and (c) and (d) mean Nash is not ESS. In short, a population of maximin could not be invaded by Nash players, but maximin could invade a population of Nash.
\begin{equation*}
u_1(M,M) > u_2(M,N),
\tag{c}
\end{equation*}
\begin{equation*}
u_1(N,N) = u_1(M,N).
\tag{d}
\end{equation*}

These illustrations, therefore, count against the generality of the Alchian and Friedman informal evolutionary argument. It is in the spirit of evolutionary arguments to treat players as having a strategy choice like play Nash equilibrium strategies or play maximin baked into them. Different strategy rules then increase/decrease in number over time depending on their relative fitness in terms of pay-offs. No player thinks about what is the best response to the other player's strategy choice at the time of the interaction. To do this would shift the analysis from evolutionary to the standard form of game theory. It is, nevertheless, interesting to consider this possibility even if the argument then ceases to be evolutionary. One can instead ask whether a maximin player loses out when compared with a player who selects a strategy that is a best response. We provide another illustration of an ordinal symmetric game where the equivalent of (a) and (b) hold (i.e.\ maximin is playing against a best responder to maximin and not a Nash equilibrium strategy player in these comparisons). Hence, even when the comparison is between maximin and a best responder, it is impossible to show that best responders always do better than maximin player.

Second, notice our result in this respect is different from a standard justification for a Nash equilibrium over maximin strategy in non-zero-sum games: Nash equilibrium is to be preferred because it gives every player at least their own maximin guarantee. The maximin guarantee is the absolute minimum pay-off received when using maximin. However, this is not necessarily the pay-off received by a player using the maximin rule as this will depend on the strategy of the other player.\footnote{Any profile of maximin strategies also yields each player at least their maximin value.} This is what we analyse. Since it is reasonable to argue that actual pay-offs matter for agents and not the guaranteed level, this is our contribution in this respect.

Finally, \citet{harsanyi1966,harsanyi1977} distinguishes a class of 2-person games, `unprofitable' ones, where the mutual use of maximin yields the same pay-off for each player as do what are the different Nash equilibrium strategies. In such games, Harsanyi defends the choice of maximin on grounds of rationality. \citet{aumann1972} similarly err on the side of maximin in such games and \citet{aumann1985} is even clearer: `under these circumstances, it is hard to see why the players would use their equilibrium strategies' (p.~668).\footnote{In addition, \citet{holler1990} and later \citet{pruzhansky2011} provide characterisations of unprofitable games in the $2\times2$ case and in 2-player games with a completely mixed Nash equilibrium, respectively. \citet{morgansefton2002} find that subjects sometimes favour maximin and sometimes Nash in such games. This literature focuses on the trade-off between the riskiness of Nash equilibrium and the safety of maximin guarantees when both yield the same pay-offs, rather than on strict Pareto comparisons between maximin profiles and Nash equilibria. A recent adjacent contribution is by \citet{ismail2025}, who introduces the concept of optimin, extending maximin to nonzero-sum games.}

Our contribution in relation to this strand of the literature comes from showing that there exist 2-player games with cardinal pay-offs in which the mutual use of maximin strictly Pareto dominates all Nash equilibria. In particular, we derive two key sets of results. 

Our first set of main results characterise how far this comparison of mutual maximin with Nash equilibria can be changed within a fixed strategic-equivalence class through conditional cardinal utility transformations (in the manner of Morris and Ui, \citeyear{morris2004}).\footnote{For ease of exposition, we focus on two-player games throughout the paper, although, as discussed in the Appendix, our main results extend substantially to $n$-player games.} For any selected Nash equilibrium, there is a strategically equivalent game in which the selected equilibrium retains its original pay-offs and not only does every maximin profile weakly Pareto dominate it, but every maximin strategy, including a pure one, guarantees those pay-offs (\cref{thm:maximin-nash-pareto}). We then fully characterise when a maximin profile can be made to strictly Pareto dominate that equilibrium while preserving its original pay-offs (\cref{thm:maximin-pareto-characterisation}). A simple sufficient condition is that both players have at least three actions and neither uses a pure strategy in the Nash equilibrium.  Finally, we fully characterise when a maximin profile can be made to weakly or strictly Pareto dominate the entire Nash equilibrium set (\cref{thm:full-weak-dominance-characterisation,cor:strict-full-dominance-characterisation}).\footnote{The Appendix also provides sufficient conditions for the reversal results, together with additional sufficiency and characterisation results on the relationships between Nash equilibrium, maximin, and minimax strategies.}

Our second set of main results consists of robustness theorems: they show that these comparisons of mutual maximin with Nash equilibria are not knife-edge. For each fixed pair of action-set sizes $(m,n)$ in the corresponding game space $\R^{2mn}$, we show the following. Apart from subsets of Lebesgue measure zero, games in which some maximin profile strictly Pareto dominates every Nash equilibrium, and games in which some Nash equilibrium strictly Pareto dominates every maximin profile, are interior points of their respective classes; hence the relevant strict ranking survives sufficiently small pay-off perturbations (\cref{thm:main-robustness}). Moreover, strict mutual maximin dominance occurs on a non-empty open, and therefore positive-measure, subset of $\R^{2mn}$ if and only if both players have at least two actions and at least one has three (\cref{thm:dimensional-robustness}). We give some reasons in Section \ref{sec:illustration} for supposing that our ‘impossibility' illustrations are not idiosyncratic and the theorems of Section 5 reinforce this conclusion.

Thus, the argument proceeds as follows. In the next section, we discuss the possible decision rules that might be associated with zero-sum thinking. Section \ref{sec:illustration} provides some illustrations of 2-person games where these decision rules do as well as Nash equilibrium strategies in the senses of (a), (b), (c) and (d) above. These illustrations form, in effect, the basis of impossibility result: i.e. it is impossible to show that Nash equilibrium strategies always do better than plausible zero-sum decision rules. Section \ref{sec:setup} introduces our setup for generalising this insight in Section \ref{sec:results}. Section \ref{sec:results} extends the existing literature on maximin: it shows that condition (a) in the generation of the impossibility result in Section~\ref{sec:illustration} is not due to the choice of some  idiosyncratic choice illustrative symmetric games Section~\ref{sec:illustration}. Instead, (a) is a property of many cardinal and ordinal games. Section~\ref{sec:conclusion} concludes.

\section{Plausible decision rules for zero-sum thinkers}
\label{sec:zero-sum-rules}

Zero-sum thinking is associated in the recent literature with seeing the world in zero-sum terms (e.g.\ see Chinoy et al., 2026). We therefore first conjecture that a zero-sum thinker will adopt a decision rule that is appropriate for zero-sum interactions and will use this rule in all their interactions. The question, therefore, arises as to what is/are the plausible decision rules that such thinkers might use.

It has been known since \citet{neumann1928} that, in 2-person zero-sum games, maximin and minimax strategies coincide and are equilibrium strategies. This insight is central to von Neumann and Morgenstern's (\citeyear{neumann1944}) classic analysis of zero-sum games. Consequently, both maximin and minimax would seem plausible potential decision rules for zero-sum thinkers to adopt. They are our first two rules.\footnote{Since then, the maximin decision rule has been studied extensively, including by \citet{wald1939}, \citet{milnor1954}, \citet{rawls1971}, and \citet{gilboa1989}; for a more recent discussion, see, for example, \citet{kuzmics2017}. Our focus, however, is not on its decision-theoretic foundations, but on how maximin performs as a behavioural rule in strategic interaction.}

There is, in addition, however, a somewhat different possible interpretation of what zero-sum thinking entails: it could be a more general rivalrous disposition that leads subjects to be concerned with their relative pay-offs. For example, the survey question in \citet{chinoy2026} that has been used to detect zero-sum thinking in the US asks subjects how strongly they agree/disagree with a statement that says `If one ethnic/US citizen/etc group becomes richer, this generally comes at the expense of other groups in the country.' This could be thought to capture a rivalrous understanding of their relationship with other that need not be confined to zero-sum games.\footnote{Likewise in the World Values Survey, subjects are asked to locate their beliefs on a Likert scale where the two extremes are given by the following two statements: `People can only get rich at the expense of others' and `Wealth can grow so there's enough for everyone'.} In this context, it seems possible that players will become concerned, not with their absolute pay-off, but with their pay-off relative to the other player as this difference captures the advantage they will enjoy over the other player (their rival).

The transformation of any 2-person game into a game of relative pay-offs, however, will always produce a zero-sum transformed game. Any pair of strategies in the initial game yielding $(u_1,u_2)$ will under this transformation become $(u_1 - u_2, u_2 - u_1)$ and hence zero-sum. Thus, the use of maximin here does not represent a change in the relation between preferences and action as such because, as is well known, maximin is rational in the expected utility maximisation sense in zero-sum games. Instead, the relative-maximin is a different decision rule because preferences are assumed to be transformed by zero-sum thinkers.\footnote{Similarly, the minimax regret criterion first transforms the original game into an auxiliary regret game against Nature and then chooses the strategy minimizing the maximum possible regret. Thus, as with relative-maximin, the zero-sum structure arises only after transforming the original game rather than from applying a zero-sum decision rule directly to the original pay-offs. See \citet{savage1951} for minimax regret and \citet{renou2010} for minimax regret equilibrium.} Since we are interested in zero-sum thinking when it changes the decision rule connecting action with preferences, we focus on the other two zero-sum rules in what follows: maximin and minimax.

We focus on maximin and minimax not only because relative-maximin captures zero sum thinking through a transformation of pay-offs rather than a change in decision rule but also because it has already been studied in relation to Nash. First, in an indirect way, relative pay-offs already appear as early as Nash's (\citeyear{nash1953}) bargaining model. Under the assumption of transferable utility, optimal threat strategies correspond to maximin strategies in the relative pay-off game (see Ismail and Peeters, \citeyear{ismail2024}). In symmetric two-player games, the relative pay-off transformation yields a symmetric zero-sum game, which underlies a distinct evolutionary literature focusing on pure strategies. \citet{schaffer1988,schaffer1989} introduce finite-population evolutionarily stable strategies, which coincide with pure relative-maximin strategies.\footnote{This differs from the earlier infinite-population notion of an evolutionarily stable strategy due to \citet{smith1973}, which is always a Nash equilibrium; for a standard textbook on this topic, see \citet{weibull1995}. An adjacent literature concerns the evolutionary stability of preferences, including contributions by \citet{bester1998}, \citet{possajennikov2000}, and \citet{alger2013}.} Subsequently, conditions under which these strategies coincide with pure Nash equilibrium strategies have been studied; see, for example, \citet{ania2008,hehenkamp2010}. In addition, \citet{duersch2012a} prove an existence result for pure equilibria in symmetric zero-sum games, which implies the existence of pure finite-population evolutionarily stable strategies.\footnote{They further show in \citet{duersch2012b} that their unbeatable imitation rule cannot be exploited indefinitely when a finite-population evolutionarily stable strategy exists in symmetric games.} In a similar spirit, \citet{ismailpeeters2025} extend this existence result and provide new sufficient conditions for pure equilibria in symmetric two-player zero-sum games, and hence for pure relative-maximin strategies in the underlying games.

We conclude by noting that our analysis of maximin and minimax as compared with Nash strategies is potentially of more general interest because there are reasons other than being a zero-sum thinker for the use of these rules. They have several desirable properties which might also explain their use, making our results of more general interest. We list these properties now to bring out why our results might have such a general interest and importance.

First, they are robust to multiplicity in a way that Nash equilibrium generally is not. If a player has several maximin strategies, any one of them guarantees at least the player's security level, regardless of which maximin strategy the opponent chooses. Nash equilibrium in non-zero-sum games does not share this robustness. When multiple equilibria exist, combining one player's strategy from one Nash equilibrium with the other player's strategy from another may fail to reproduce any Nash equilibrium pay-off and can even yield outcomes strictly below each player's security level. To be sure that the pay-off from playing Nash is as good in these circumstances as the security level achieved by maximin requires players to coordinate on the same Nash equilibrium (see Harsanyi, 1964).

Second, the informational requirements of these decision rules  are lower. Nash equilibrium reasoning requires players to have sufficient knowledge of their opponents' pay-offs as well as their own. In contrast, minimax requires knowledge of the opponent's pay-offs  and maximin depends only on knowledge of one's own pay-offs.

Third, these rules make less computational demands. In finite two-player games, maximin and minimax can both be computed efficiently using linear programming. By contrast, there is no known polynomial-time algorithm for finding mixed Nash equilibria even in 2-player games \citep{daskalakis2006,chen2006}; moreover, determining whether equilibria satisfy natural properties is NP-hard \citep{gilboa1989b}. While this computational contrast does not by itself establish behavioural plausibility, it reinforces the idea that these zero-sum rules may serve as more tractable and hence more natural decision procedures for agents facing strategic complexity.

Fourth, these differences extend to existence properties under weaker pay-off assumptions. When attention is restricted to pure strategies and only ordinal (rather than cardinal) pay-offs are assumed, zero-sum decision rules continue to guarantee the existence of optimal strategies in pure form. By contrast, it is well known that pure strategy Nash equilibria do not always exist.

Since any of these reasons might explain the adoption of maximin or minimax rather than Nash, it is more generally interesting than the question of zero-sum thinking to know what are  the pay-off properties of maximin and minimax as compared with Nash.

\section{Some impossibility illustrations}
\label{sec:illustration}

In this section, we give four illustrations of games that confound the background evolutionary argument that Nash equilibrium strategies associated with expected utility maximisation always produce better results for players than zero-sum rules. In effect, these are the basis for an impossibility result: it is impossible to show that the adoption of zero-sum rules will always harm players relative to their use of Nash equilibrium strategies. We show in later sections that at least with respect to one of the key advantages of maximin over Nash equilibrium, this advantage is shared by maximin in many games. In this sense, there is no reason to suppose these illustrations are idiosyncratic.

Our first illustration is the $3 \times 3$ symmetric game in Figure~\ref{fig:3x3_main}.
\begin{figure}[htbp]
\centering
	\[
	\begin{array}{c@{~~} | @{~~}c@{~~}c@{~~}c}
		& x & y & z\\ \hline
		x & 7,7 & 9,4 & 0,8\\
		y & 4,9 & 8,8 & 7,4\\
		z & 8,0 & 4,7 & 3,3
	\end{array}
	\]
\caption{A symmetric $3\times3$ game where the maximin profile strictly outperforms Nash equilibrium}
\label{fig:3x3_main}
\end{figure}

This game has a unique mixed Nash equilibrium $s^N = [(1/2, 1/4, 1/4),(1/2, 1/4, 1/4)]$, a unique maximin strategy $s_i^M = (0, 5/8, 3/8)$, and a unique minimax strategy $s_i^m = (1/2, 0, 1/2)$ for each player $i \in \{1,2\}$.\footnote{The maximin strategy guarantees a pay-off of 5.5, whereas the Nash strategy guarantees only 2.5 if player 2 plays $z$, even though $z$ is a best response to player 1's Nash strategy.} Figure \ref{fig:3x3_induced} gives the induced game pay-offs for each type of rule follower when they interact with any of the three possible types of rule follower.

\begin{figure}[htbp]
\centering
\[
\begin{array}{c@{~~} | @{~~}ccc}
 & s_2^N & s_2^M & s_2^m\\ \hline
s_1^N & (5.75, 5.75) & (5.625, 5.75) & (4.5, 5.75)\\
s_1^M & (5.75, 5.625) & (6.125, 6.125) & (5.5, 4.625)\\
s_1^m & (5.75, 4.5) & (4.625, 5.5) & (4.5, 4.5)
\end{array}
\]
\caption{The induced game pay-offs by each rule}
\label{fig:3x3_induced}
\end{figure}

If zero-sum thinking is associated with maximin, then in a population of Nash equilibrium and maximin players condition (a) holds ($6.125>5.75$) and so does condition (b) ($5.75>5.625$). Thus on a strict comparison of the difference in rule where the only difference in the comparison of the rules are the rules themselves, maximin players do better than Nash equilibrium ones.\footnote{Similarly, relative-maximin players also outperform Nash equilibrium players, yielding pay-offs of 5.95 for each player.} Furthermore condition (c) holds ($6.125>5.625$) and so does condition (d) and so maximin is an ESS in the induced game whereas Nash equilibrium is not.

If zero-sum thinking is instead associated with minimax, then while the equivalent of (a) does not hold, (b) does for minimax over Nash. Further if the population comprised of only the two types of zero-sum players, then conditions (a), (b), (c) and (d) hold in favour of maximin over minimax.

Hence, in the Figure \ref{fig:3x3_main} game, we would expect maximin players to outperform both Nash equilibrium ones and minimax ones.

Our next illustration is the symmetric $3 \times 3$ game with ordinal pay-offs in Figure \ref{fig:3x3_ordinal_2532}, focusing only on pure strategies.

\begin{figure}[htbp]
\centering
\[
\begin{array}{c@{~~} | @{~~}c@{~~}c@{~~}c}
 & x & y & z\\ \hline
x & 1,1 & 2,9 & 7,8\\
y & 9,2 & 5,5 & 3,4\\
z & 8,7 & 4,3 & 6,6
\end{array}
\]
	\caption{A symmetric $3\times3$ game with ordinal pay-offs in which the maximin profile strictly Pareto dominates the unique pure Nash equilibrium}
	\label{fig:3x3_ordinal_2532}
\end{figure}

In this game, the Nash equilibrium is given by $(y,y)$ and mutual maximin yields $(z,z)$. Thus, in a population of maximin and Nash equilibrium players, i.e., in the induced $2\times2$ game between $y$ and $z$ strategies, condition (a) holds ($6>5$) and (b) holds ($4>3$). So, in the strict comparison of the strategies, maximin is again superior to the Nash equilibrium. Further (c) holds ($6>3$, so maximin is ESS) but (d) does not ($4<5$) and the strict inequality here makes Nash an ESS too. We introduce this example for several reasons.

First, the example clarifies a limitation of standard evolutionary reasoning in positive-sum environments for the purposes of answering our question. In the full $3\times3$ game, $y$ is the only ESS and not $z$ because $x$, when rare, can invade $z$: it receives 7 against $z$, whereas most $z$-types will be interacting with $z$ types when $x$ is rare and so they receive 6. But this means only that the share of $x$-types increases from rarity. It does not mean that $z$ players are harmed in the process since a $z$-player matched with $x$ gets 8, which is better than the 6 they get when matched with a fellow $z$ player.

Second, the example highlights a difference between our notion of evolutionary success and Schaffer's finite-population ESS. Under this notion, $z$ remains stable because it defeats $x$ and $y$ in pairwise comparison: $8>7$ against $x$ and $4>3$ against $y$. However, this conclusion is independent of the absolute level of the pay-off at $(z,z)$. If, for example, the outcome for $(z,z)$ were replaced by a much lower outcome such as $(1,1)$, $z$ would still be finite-population ESS but it would not satisfy (a); and so the use of $z$ could in these circumstances be said to harm an agent's interests.

Third, it is natural in addressing an evolutionary argument to assume that players are associated with playing strategies like `maximin' or `Nash equilibrium'. They do not think about the strategy as such, they are more simply bearers of this way of behaving. The analysis then proceeds on the basis of whether one or other player/strategy prospers more than the other with the result that one type of player/strategy increases in number over the other (or not as the case may be). Strategies are baked into the players and evolution selects those that prove the fittest. However, this evolutionary setting, although it is the basis of Alchian and Friedman argument that we are addressing, is not the only one where it is interesting to ask whether playing maximin damages a player's interests. In conventional game theory, the appropriate comparison for the maximin player would be with someone who always `best responds'. With a fellow `best responder', this will plausibly lead players to select Nash equilibrium strategies. Hence if (a) holds in the evolutionary setting, the equivalent of (a) will also hold in the conventional one where the comparison becomes between maximin and `best response'. Condition (b) however does not automatically carry over in this way for mixed interactions because the best response to maximin need not be to play the Nash equilibrium strategy. Indeed, in the Figure \ref{fig:3x3_ordinal_2532} game, row's best response to column's maximin is `$x$' and not the Nash equilibrium strategy of `$y$'. Nevertheless, and this is the point of this illustration, the equivalent of (b) still holds: the return in this mixed player interaction again favours maximin over the best responder ($8>7$). Thus, this illustration provides the confound to any general proposition that maximin damages the interests of those who use it in positive-sum games because in a conventional game theory setting maximin still does better than a `best responder' in the sense of (a) and (b).\footnote{Although it does not make much sense to apply an evolutionary game theoretic equilibrium concept to a game where one type of player is not represented in an evolutionary fashion, neither maximin nor a best response rule are evolutionarily stable strategies in these circumstances.}

Our third illustration of where action guided by Nash is (potentially) injurious compared with maximin is given by the game in Figure \ref{fig:5x5}. For the row player, the unique maximin strategy is the pure strategy `$x$', which guarantees 2. For column player, `$x$' is the pure maximin strategy (although there are mixed maximin strategies as well), and it also guarantees 2. The pure maximin profile $(x, x)$ yields pay-off $(3, 12)$. In contrast, the game has exactly two Nash equilibria, both in pure strategies: $(y, y)$ and $(z, z)$, with pay-offs $(4, 2)$ and $(3, 5)$, respectively. $(3,12)$ from the joint maximin is Pareto superior to the second of these but not the first. However, Nash players have a coordination problem to solve if they are to achieve one or other of the Nash equilibria. If they do not solve this problem, then playing Nash will on occasion produce the non-Nash outcomes of $(y,z)$ and $(z,y)$; and both these outcomes yield $(0,0)$.

\begin{figure}[htbp]
\centering
\[
\begin{array}{c@{~~} | @{~~}c@{~~}c@{~~}c@{~~}c@{~~}c}
 & x & y & z & v & w\\ \hline
x & 3,12 & 3,2 & 2,1 & 7,5 & 2,6\\
y & 1,2 & 4,2 & 0,0 & 3,0 & 2,0\\
z & 4,3 & 0,0 & 3,5 & 9,1 & 1,4\\
v & 2,4 & 4,3 & 1,5 & 0,6 & 1,1\\
w & 0,5 & 3,7 & 0,3 & 0,8 & 3,6
\end{array}
\]
\caption{A game with two pure Nash equilibria in which average Nash-play pay-offs are below the maximin guarantees}
\label{fig:5x5}
\end{figure}

In particular, suppose each player randomises over the two Nash equilibria strategies with the result that the two Nash equilibria are equally likely and as likely as the two non-Nash outcomes. The average pay-off from Nash is

\begin{equation*}
\frac{1}{4}\bigl[(4,2) + (0,0) + (0,0) + (3,5)\bigr]
= \left(\frac{7}{4}, \frac{7}{4}\right)
= (1.75,\,1.75).
\end{equation*}

Since $1.75 < 2$, this average pay-off lies strictly below each player's maximin guarantee. Hence the analogue of condition (a) will hold in these circumstances. One attraction of maximin is precisely its robustness to multiplicity: no matter which maximin strategy the opponent chooses, the player still receives at least the security level. Nash equilibrium does not have this property. When there are multiple equilibria, its pay-off predictions depend on successful coordination across players, and when that coordination fails, realised pay-offs can fall below what maximin guarantees. In that sense, the attraction of maximin is not limited to games such as Figures \ref{fig:3x3_main} or \ref{fig:3x3_ordinal_2532}, where maximin strictly Pareto dominates Nash equilibrium. It also extends to games like Figure \ref{fig:5x5}, where the difficulty of implementing Nash equilibrium can itself make maximin not only safer but also, in realised pay-off terms, the superior rule.

One possible explanation of how the contrary background or received view might have gained hold is that matters are rather different in $2\times2$ games, symmetric games like the Prisoners' Dilemma, Hawk-Dove, etc.
Indeed in Appendix~\ref{app:robustness-theorems}, \cref{lemma:2x2-maximin-no-strict-pareto} we show that it is strictly impossible for mutual maximin to Pareto dominate all Nash equilibria. For example, in Hawk-Dove, mutual maximin strictly dominates the mixed Nash equilibrium but not the two pure equilibria (see the Appendix for this and other illustrations). In turn, this might fuel the suspicion that our counter illustrations in this section are somewhat idiosyncratic. We show this not to be the case in Section~\ref{sec:results}. For now, however, we also note that  it is relatively easy to conduct a statistical analysis of all the symmetric $3\times 3$ ordinal games like that of Figure~\ref{fig:3x3_ordinal_2532}. In particular, we show in the Appendix that on balance across all these games, neither Nash nor maximin is superior in the sense of (a). In the same vein, the Appendix contains some familiar economic applications where the relative superiority of maximin or Nash behaviour depends on the parameters of the game (e.g. in Cournot model with capacity constraints). 

Towards this same end, we conclude this section with an asymmetric $2\times2$ game, where the literature has already raised a doubt about the superiority of Nash over maximin.
	\[
\begin{array}{c@{~~} | @{~~}c@{~~}c}
		& x & y\\ \hline
		x & 3,1 & 0,2\\
		y & 1,4 & 2,0
	\end{array}
	\]
\citet{harsanyi1964} discusses this `unprofitable' game. It has a unique mixed Nash equilibrium, $((4/5, 1/5), (1/2, 1/2))$, yielding pay-offs $(3/2, 8/5)$. Player 1's maximin strategy is $(1/4, 3/4)$, and player 2's maximin strategy is $(2/5, 3/5)$. These strategies guarantee pay-offs of $3/2$ and $8/5$, respectively, so each player's security level is exactly equal to that player's Nash equilibrium pay-off. Moreover, the resulting maximin profile yields the same pay-off vector $(3/2, 8/5)$.\footnote{A $3\times3$ game with similar properties can be constructed in which there exists a pure strategy Nash equilibrium and a pure maximin profile that generate the same pay-offs \citep[p.~62]{hargreaves_heap2004}.}

Thus, each player secures the Nash equilibrium pay-off by using a maximin strategy, whereas the corresponding Nash strategy does not secure that pay-off against all responses. For example, player 1's Nash strategy yields only 0.4 if player 2 deviates to $y$. The literature on unprofitable games largely stops at this point: it asks when equilibrium pay-offs coincide with maximin guarantees. The question we now take up is broader: what are the conditions under which zero-sum rules do strictly better than Nash equilibrium ones?

\section{Setup}
\label{sec:setup}
Let
	 $
	G=\langle N,(A_i)_{i\in N},(u_i)_{i\in N}\rangle
	$
	be a finite two-player normal-form game, where $N=\{1,2\}$ is the set of players, $A_i$ is player $i$'s finite action set, and $u_i\colon A_1\times A_2\to\R$ is player $i$'s pay-off function. As usual, $-i$ denotes the opponent of player $i$.\footnote{For a standard textbook on game theory, see, e.g., \citet{osborne1994}.}
	
	Let $S_i$ denote the simplex of mixed strategies on $A_i$. A mixed-strategy profile is $s=(s_1,s_2)\in S_1\times S_2:= S$, and expected pay-offs are defined by bilinear extension: for every $i\in N$,
	$u_i(s_1,s_2)
	=
	\sum_{a_1\in A_1}\sum_{a_2\in A_2}
	s_1(a_1)s_2(a_2)u_i(a_1,a_2)$.
	For $s_i\in S_i$, its support is
	$
	\supp(s_i):=\{a_i\in A_i: s_i(a_i)>0\}.
	$
	For any mixed profile $s$, we also write
	$
	u(s):=(u_1(s),u_2(s))
	$
	for the associated pay-off vector. Because the action sets are finite, all simplices are compact and every minimum and maximum shown below is attained.
	
	\begin{definition}[Pareto comparisons]
		A strategy profile $s$ \emph{weakly Pareto dominates} a profile $t$ if
		 $
		u_i(s)\ge u_i(t)
		$ for all $i\in N$.
		It \emph{Pareto dominates} $t$ if it weakly Pareto dominates $t$ and the inequality is strict for at least one player. It \emph{strictly Pareto dominates} $t$ if
		 $
		u_i(s)>u_i(t)
		$ for all $i\in N$.
		A profile $s$ is \emph{Pareto optimal} if there is no profile $t\neq s$ that Pareto dominates $s$.
	\end{definition}
	
	\begin{definition}[Best response and Nash equilibrium]
		For $s_{-i}\in S_{-i}$, player $i$'s mixed best response set is
		 $
		\BR_i(s_{-i})
		:=
		\argmax_{t_i\in S_i}u_i(t_i,s_{-i}).
		$
		A mixed profile $s^{N}=(s_1^{N},s_2^{N})$ is a \emph{Nash equilibrium} if
		 $
		s_i^{N}\in \BR_i(s_{-i}^{N})
		$ for all $i\in N$.
        We write $\NE(G)$ for the set of all Nash equilibria of the game $G$.
	\end{definition}
	
	A Nash equilibrium is thus a profile of mutually optimal replies: once the opponent's equilibrium strategy is fixed, no player can improve by deviating unilaterally.
	
	\begin{definition}[Maximin and minimax]
		The \emph{maximin value}, or security level, of player $i$ is
		\[
		\underline v_i:=\max_{s_i\in S_i} \min_{t_{-i}\in S_{-i}}u_i(s_i,t_{-i}),
		\]
		and the \emph{maximin set} is
		$
		M_i:=\argmax_{s_i\in S_i} \min_{t_{-i}\in S_{-i}}u_i(s_i,t_{-i}).
		$
		Any strategy $s^M_i\in M_i$ is called a \emph{maximin strategy} of player $i$. A profile $s^M=(s^M_1,s^M_2)\in M_1\times M_2$ is called a \emph{maximin profile}. 
		
		The \emph{minimax value} of player $i$ is $
		\overline v_i:=
		\min_{s_{-i}\in S_{-i}}
		\max_{t_i\in S_i}u_i(t_i,s_{-i})$.
		A \emph{minimax strategy} of player $1$ is any strategy $s^m_1$ such that $s^m_1\in
		\argmin_{s_1\in S_1}
		\max_{t_2\in S_2}u_2(s_1,t_2)$,
		and a \emph{minimax strategy} of player $2$ is any $s^m_2$ such that $s^m_2\in
		\argmin_{s_2\in S_2}
		\max_{t_1\in S_1}u_1(t_1,s_2)$.
		A profile $s^m=(s^m_1,s^m_2)$ is called a \emph{minimax profile} if each component is a minimax strategy for the corresponding player.
	\end{definition}

	Maximin and minimax solve different optimization problems. Maximin is self-protection against an adversarial opponent: it asks what pay-off a player can guarantee for herself against every opponent response.\footnote{Hence, we use the terms maximin value, maximin guarantee, and security level interchangeably throughout.} Minimax is an opponent-induced cap: it asks how low a player can force the opponent's best attainable pay-off.
	
	\subsection{Basic equalities and inequalities}
	\label{subsec:basic}
	The next remarks collect well-known results. Short proofs are generally included for completeness.
	
	\begin{remark}
		\label{rem:maximin-minimax}
		The minimax theorem, due to \citet{neumann1928}, states that for each player $i$:
		\[
		\max_{s_i}\min_{s_{-i}}u_i(s_i,s_{-i})
		=
		\min_{s_{-i}}\max_{s_i}u_i(s_i,s_{-i}).
		\]
		That is, $\underline v_i = \overline v_i = v_i$, where $v_i$ is called the value of the zero-sum game with pay-offs $(u_i,-u_i)$.
	\end{remark}
	
	Although the game $G$ is general-sum, each player's maximin and minimax problems are defined using the same pay-off function $u_i$ and therefore form a zero-sum game between the player and an adversarial opponent. Von Neumann's minimax theorem implies that, in mixed strategies, the player's guaranteed pay-off (maximin) coincides with the smallest pay-off the opponent can enforce against the player's best response (minimax). Note that the strategies achieving these values are player $i$'s maximin
	strategy $s_i^{M}$ and the opponent's minimax strategy $s_{-i}^{m}$ (with respect to $u_i$).
	
\begin{remark}[Equilibrium and maximin pay-offs lie above security levels]
    \label{rem:NE-above-security}
    For every Nash equilibrium $s^{N}$ and each player $i$,
    $
    u_i(s^{N}) \ge v_i.
    $
For every maximin profile $s^M = (s_1^M, s_2^M) \in M_1 \times M_2$ and each player $i$, 
$
u_i(s^M) \ge v_i.
$
\end{remark}

Two important points to note here. First, the value $v_i$ is a worst-case guarantee. It need not coincide with the pay-offs $u_i(s^M_1, s^M_2)$ at a particular maximin profile $(s^M_1, s^M_2)$. Thus, \cref{rem:NE-above-security} compares equilibrium pay-offs with security levels, not with the pay-offs at a maximin profile, which may lie strictly above security levels. Moreover, a maximin profile may strictly Pareto dominate a Nash equilibrium. 

Second, note the interchangeability of maximin strategies in the sense that no matter which maximin strategy a player chooses each receives at least their security level from the maximin profile. However, it is well known that both players may end up strictly below their security level, for example, in coordination games, by choosing a Nash equilibrium strategy.
    
	\begin{proof}
		Fix $i$ and choose any $s^M_i\in M_i$. Because $s^{N}$ is a Nash equilibrium, $s_i^{N}$ is a best response to $s_{-i}^{N}$. Therefore
		$
		u_i(s^{N})=u_i(s_i^{N},s_{-i}^{N})
		\ge u_i(s^M_i,s_{-i}^{N}).
		$
		Now use the maximin property of $s^M_i$:
		$
		u_i(s^M_i,s_{-i}^{N})
		\ge \min_{t_{-i}}u_i(s^M_i,t_{-i})
		=v_i.
		$
		Combining the two inequalities gives $u_i(s^{N})\ge v_i$. As for the maximin profile, since a maximin strategy guarantees $v_i$ for each player $i$, any combination of maximin strategies must also yield at least $v_i$.
	\end{proof}
	
	\begin{remark}
		\label{rem:minimax-component}
		Let $s^m_1$ be a minimax strategy of player $1$. Then for every $t_2\in S_2$ and every Nash equilibrium $s^{N}$,
		$
		u_2(s^m_1,t_2)\le u_2(s^{N}).
		$
		Symmetrically, if $s^m_2$ is a minimax strategy of player $2$, then for every $t_1\in S_1$ and every Nash equilibrium $s^{N}$,
		$
		u_1(t_1,s^m_2)\le u_1(s^{N}).
		$
		Consequently, a profile with one minimax component cannot strictly Pareto dominate any Nash equilibrium.
	\end{remark}
	
	A minimax strategy caps the opponent's pay-off. Once one component is minimax, the opponent's pay-off at any profile using it is already bounded above by what that opponent gets at any equilibrium, so strict Pareto improvement over equilibrium is impossible.
	
	\begin{proof}
		Suppose first that $s^m_1$ is a minimax strategy of player $1$. By definition,
		$
		\max_{t_2\in S_2}u_2(s^m_1,t_2)=v_2.
		$
		Hence, for every $t_2\in S_2$,
		$
		u_2(s^m_1,t_2)\le v_2.
		$
		Now let $s^{N}=(s_1^{N},s_2^{N})$ be a Nash equilibrium. It implies that
		$
		u_2(s^{N})=\max_{t_2\in S_2}u_2(s_1^{N},t_2).
		$
		Therefore
		$
		u_2(s^{N})
		\ge
		\min_{s_1\in S_1}\max_{t_2\in S_2}u_2(s_1,t_2)
		=v_2.
		$
		Combining the two inequalities yields
		$
		u_2(s^m_1,t_2)\le v_2\le u_2(s^{N}).
		$
		The second statement is symmetric.
		
		Finally, if a profile has one minimax component, then one player's pay-off cannot exceed the corresponding equilibrium pay-off. So such a profile cannot strictly Pareto dominate any Nash equilibrium.
	\end{proof}

	\begin{corollary}[Nash equilibrium dominates minimax profile]
		\label{cor:minimax-profile-dominated}
		Let $s^m$ be a minimax profile. Then every Nash equilibrium $s^{N}$ weakly Pareto dominates $s^m$.
	\end{corollary}	
	\begin{proof}
		Apply \cref{rem:minimax-component} with the minimax strategy $s^m_1$ of player $1$ and set $t_2=s^m_2$ to obtain
		$
		u_2(s^m_1,s^m_2)\le u_2(s^{N}).
		$
		Applying the symmetric part of the same remark with the minimax strategy $s^m_2$ of player $2$ and setting $t_1=s^m_1$ gives
		$
		u_1(s^m_1,s^m_2)\le u_1(s^{N}).
		$
		Thus $s^{N}$ weakly Pareto dominates $s^m$.
	\end{proof}

	\begin{remark}
		\label{rem:one-NE-strategy}
		Let $s^N=(s^N_1,s^N_2)$ be a Nash equilibrium and $t\in S$ be any strategy profile. Then, no profile of the form $(s^N_i,t_{-i})$ can strictly Pareto dominate $s^N$.
		
		\noindent
		This is because, by the definition of Nash equilibrium, for every $t_2\in S_2$,
		$
		u_2(s^N_1,t_2)\le u_2(s^N).
		$
		The symmetric statement holds for player 1 and profiles of the form $(t_1,s^N_2)$. Therefore, no such ``hybrid'' profile can strictly Pareto dominate the Nash equilibrium.
	\end{remark}
	
	Once one player's equilibrium strategy is fixed, the opponent's equilibrium response already maximises her pay-off against it. So changing only the opponent's strategy cannot make both players strictly better off relative to that same equilibrium.

\section{Results}
\label{sec:results}

\subsection{Nash-maximin Pareto relationships}
The results in this subsection first characterise when strategic equivalence can make a maximin profile match or dominate a selected Nash equilibrium, and when it can dominate the entire Nash equilibrium set.

\begin{definition}[Strategic equivalence]
\label{def:strategic-equivalence}
Two games $G$ and $\widetilde G$ with the same players and action sets are \emph{strategically equivalent} if, for each player $i$, there is a function $h_i\colon A_{-i}\to\R$ such that $\widetilde u_i(a_i,a_{-i})=u_i(a_i,a_{-i})+h_i(a_{-i})$
for every pure action profile $(a_i,a_{-i})$.
\end{definition}

The added term depends on the opponent's action but not on the player's own action. Consequently, strategic equivalence preserves each player's preferences over their own pure and mixed strategies conditional on every mixed strategy of the opponent. In particular, strategically equivalent games have the same better-response and best-response correspondences and hence the same Nash equilibrium set. This additive form is the nonstrategic component identified by \citet{candogan2011}, is a special case of VNM-equivalence in \citet{morris2004}, and is also studied by \citet{abdou2022}.

\begin{lemma}[Strategic equivalence preserves Nash equilibrium \citep{morris2004}]
\label{lem:strategic-equivalence}
Strategically equivalent games have the same best-response correspondences and the same Nash equilibrium set.
\end{lemma}

This follows from \citet[pp.~263--264]{morris2004}: the transformation in \cref{def:strategic-equivalence} is the unit-scale case ($w_i=1$) of VNM-equivalence and therefore preserves the full best-response structure and the entire Nash-equilibrium set. It need not, however, preserve equilibrium pay-offs, maximin strategies, or security levels. The constructions in this section are designed to preserve some or all of these additional properties.

\begin{theorem}[Maximin-Nash Pareto dominance]
\label{thm:maximin-nash-pareto}
Let $G$ be a finite two-player game and fix $\bar s\in\NE(G)$.
There exists a game $\widehat G$ strategically equivalent to $G$ such
that:
\begin{enumerate}[label=\textup{(\roman*)}]
\item for every $i\in N$, the selected equilibrium pay-offs and the security levels coincide:
$$\widehat u_i(\bar s)=v_i(\widehat G)=u_i(\bar s);$$
\item every maximin profile
$m\in M_1(\widehat G)\times M_2(\widehat G)$ weakly Pareto dominates
the selected equilibrium;
\item  there is a pure maximin profile $a^M\in M_1(\widehat G)\times M_2(\widehat G)$ such that
$\widehat u_i(a^M)=u_i(\bar s)$ for all $i\in N$.
\end{enumerate}
\end{theorem}

\noindent\emph{Proof.} See Appendix~\ref{app:selected-equilibrium}.

This theorem says that, for any selected Nash equilibrium, one can preserve its original pay-off vector while making those pay-offs equal to both players' security levels. Every maximin profile then weakly Pareto dominates the selected equilibrium, and at least one pure maximin profile yields exactly the same pay-off vector. Strict Pareto dominance is not always attainable; for example, it is impossible in a game in which all pay-offs are equal. The next theorem characterises exactly when a maximin profile can instead be made to strictly Pareto dominate the selected equilibrium.

\begin{example}
\label{ex:maximin-nash-pareto-construction}

To start from a familiar example, consider the game in Figure~\ref{fig:3x3_main}. Multiply its payoffs by 4 and add
\(14\), \(-4\), and \(0\) to player~1's pay-offs in columns \(x\),
\(y\), and \(z\), respectively, with the symmetric adjustments for
player~2. This gives the game in Figure~\ref{fig:thm:maximin-nash-pareto}~\textup{(a)}.

	\begin{figure}[htbp]
		\centering
		\[
		\begin{array}{c@{\qquad\qquad}c}
			\begin{array}{c@{~~} | @{~~}c@{~~}c@{~~}c}
				& x & y & z\\ \hline
				x & 42,42 & 32,30 & 0,46\\
				y & 30,32 & 28,28 & 28,12\\
				z & 46,0  & 12,28 & 12,12
			\end{array}
			&
			\begin{array}{c@{~~} | @{~~}c@{~~}c@{~~}c}
				& x & y & z\\ \hline
				x & 41,41 & 33,29 & 1,45\\
				y & 29,33 & 29,29 & 29,13\\
				z & 45,1  & 13,29 & 13,13
			\end{array}
			\\[0.5em]
			\textup{(a) }G
			&
			\textup{(b) }\widehat G
		\end{array}
		\]
		\caption{An illustration of \cref{thm:maximin-nash-pareto}}
		\label{fig:thm:maximin-nash-pareto}
	\end{figure}
	
To apply \cref{thm:maximin-nash-pareto} to this particular game, add
$(-1,1,1)$ to player~1's pay-offs in columns $(x,y,z)$,
respectively, and make the symmetric adjustments to player~2. These
game-specific adjustments are chosen so that their expectation under
the equilibrium strategy $(1/2,1/4,1/4)$ is zero, thereby preserving
the selected equilibrium pay-off, while making the pure action $y$
yield the constant pay-off (29) against every opponent action. The
resulting game is shown in
Figure~\ref{fig:thm:maximin-nash-pareto}\textup{(b)}.

Because the adjustments depend only on the opponent's action, they
preserve all best-response comparisons. Hence the transformed game is
strategically equivalent to the game in
Figure~\ref{fig:3x3_main} and has the same unique Nash equilibrium,
$
\bar s
=
\left[
\left(\frac12,\frac14,\frac14\right),
\left(\frac12,\frac14,\frac14\right)
\right].
$

In \(G\), the equilibrium pay-off is
$
u_i(\bar s)=29.
$
The unique maximin strategy is the pure action \(y\), which guarantees
\(28\), and the mutual-maximin profile yields
$
u_i(y,y)=28.
$
In passing from \(G\) to \(\widehat G\), the expected
adjustment at the equilibrium is zero:
$
\frac12(-1)+\frac14(1)+\frac14(1)=0.
$
Thus
$
\widehat u_i(\bar s)=29.
$
Moreover, in \(\widehat G\), action \(y\) yields the constant pay-off
of 29. 
It is the unique maximin strategy and guarantees exactly
the equilibrium pay-off.
Consequently,
$
\widehat u(y,y)
=
\widehat u(\bar s)
=
(29,29).
$

Thus Nash equilibrium strictly Pareto dominates the maximin profile
in \(G\), whereas the construction in
\cref{thm:maximin-nash-pareto} makes the pure maximin strategy
guarantee exactly the selected equilibrium pay-off in \(\widehat G\).
\end{example}

\begin{theorem}[Maximin-Nash strict Pareto dominance characterisation]
\label{thm:maximin-pareto-characterisation}
Let $G$ be a finite two-player game, fix $\bar s\in\NE(G)$, and fix a candidate profile $m\in S_1\times S_2$. The following are equivalent:
\begin{enumerate}[label=\textup{(\alph*)}]
\item there is a strategically equivalent game $\widetilde G$ such that $m$ is a maximin profile and
$\widetilde u_i(m)>\widetilde u_i(\bar s)$ for all $i\in N$;

\item for each player $i$, there are a belief $\mu_{-i}\in S_{-i}$ and non-negative pay-off additions $r_i(a_{-i})\geq0$, one for each $a_{-i}\in A_{-i}$, such that
\begin{equation}
m_i\in\BR_i^G(\mu_{-i}),
\qquad
\sum_{a_{-i}\in A_{-i}}\mu_{-i}(a_{-i})r_i(a_{-i})=0,
\label{eq:selected-support-addition}
\end{equation}
and
\begin{equation}
\sum_{a_{-i}\in A_{-i}}
m_{-i}(a_{-i}) r_i(a_{-i})>\sum_{a_{-i}\in A_{-i}}\bar s_{-i}(a_{-i}) r_i(a_{-i}).
\label{eq:selected-separation}
\end{equation}
\end{enumerate}
Whenever these conditions hold, then for every pair $\ell_1,\ell_2>0$ the \(\widetilde G\) can additionally be chosen so that
$
\widetilde u_i(\bar s)=u_i(\bar s)$, 
$m_i\in M_i(\widetilde G)$, and 
$\widetilde u_i(m)>u_i(\bar s)+\ell_i$ for every $i\in N$.
\end{theorem}
\noindent\emph{Proof.} See Appendix~\ref{app:selected-equilibrium}.

The characterisation separates the two roles of the opponent-dependent adjustment. First, \(m_i\) must be a best response to some belief \(\mu_{-i}\). The non-negative pay-off additions have expectation zero under this belief, so every opponent action in its support receives zero addition. These actions can therefore pin down the security level while the original best-response comparison ensures that no other strategy can guarantee more. Second, the additions must favour the opponent's strategy $m_{-i}$ more than the equilibrium strategy $\bar s_{-i}$, as stated in condition \eqref{eq:selected-separation}.

We next provide two simple sufficient conditions for the characterisation. The first applies to a given pure candidate profile: each candidate action must be a best response to some different pure opponent action, while the selected equilibrium must not place probability one on the candidate opponent action. The second shows that such a pure candidate profile necessarily exists when both players mix at the selected equilibrium and the game has at least the $2\times 3$ dimensional threshold.

\begin{corollary}
\label{cor:pure-selected-separation}
Fix a Nash equilibrium \(\bar s\) and a pure candidate profile \(a^M=(a_1^M,a_2^M)\). Suppose that, for each player \(i\), there is a pure opponent action \(b_{-i}\neq a_{-i}^M\) such that
$a_i^M\in\BR_i^G(b_{-i})$ and  $\bar s_{-i}(a_{-i}^M)<1.$
Then there is a strategically equivalent game in which \(a^M\) is a mutual-maximin profile and strictly Pareto dominates \(\bar s\). The selected equilibrium pay-off can be preserved, and the dominance margin can be made arbitrarily large. If \(a_i^M\) is the unique best response to \(b_{-i}\) for each player, then the constructed maximin strategies are unique.
\end{corollary}

\noindent\emph{Proof.} See Appendix~\ref{app:selected-equilibrium}.

\begin{corollary}
\label{cor:mixed-selected-separation}
Suppose each player has at least two actions, at least one player has at least three actions, and the selected Nash equilibrium \(\bar s\) satisfies $|\supp(\bar s_i)|\geq2$ for every $i\in N$. 
Then, there is a strategically equivalent game with a pure mutual-maximin profile that strictly Pareto dominates \(\bar s\) by any prescribed positive margins, while preserving the selected equilibrium pay-off vector.
\end{corollary}

\noindent\emph{Proof.} See Appendix~\ref{app:selected-equilibrium}.

\begin{theorem}[Maximin-all-Nash Pareto dominance characterisation]
\label{thm:full-weak-dominance-characterisation}
Fix a candidate profile $m\in S_1\times S_2$. There is a strategically equivalent game $\widetilde G$ such that for every $i\in N$ and $s\in\NE(G)$, 
$m_i\in M_i(\widetilde G)$ and $ \widetilde u_i(m)\geq\widetilde u_i(s)$  if and only if, for each player $i$, there are a belief $\mu_{-i}\in S_{-i}$ and non-negative pay-off additions $r_i(a_{-i})\geq0$, one for each $a_{-i}\in A_{-i}$, such that
\begin{equation}
m_i\in\BR_i^G(\mu_{-i}),
\qquad
\sum_{a_{-i}\in A_{-i}}\mu_{-i}(a_{-i})r_i(a_{-i})=0,
\label{eq:full-support-addition}
\end{equation}
and
\begin{equation}
\sum_{a_{-i}\in A_{-i}}
\bigl(m_{-i}(a_{-i})-s_{-i}(a_{-i})\bigr)r_i(a_{-i})
\geq u_i(s)-u_i(m_i,s_{-i})
\qquad
\text{for every }s\in\NE(G).
\label{eq:full-cardinal-inequality}
\end{equation}
Whenever these conditions hold, the security levels may be prescribed arbitrarily. The game $\widetilde G$ can be chosen so that, at every $s\in\NE(G)$,
\begin{equation}
\widetilde u_i(m)-\widetilde u_i(s)
=
\sum_{a_{-i}\in A_{-i}}
\bigl(m_{-i}(a_{-i})-s_{-i}(a_{-i})\bigr)r_i(a_{-i})
-\bigl[u_i(s)-u_i(m_i,s_{-i})\bigr].
\label{eq:full-difference}
\end{equation}
\end{theorem}
\noindent\emph{Proof.} See Appendix~\ref{app:full-weak-dominance-characterisation}.

Here the only additional requirement compared with \cref{thm:maximin-pareto-characterisation} is that the same construction must work against every Nash equilibrium, rather than against one selected equilibrium. Thus, for each $s\in\NE(G)$, the expected pay-off addition advantage of $m_{-i}$ over $s_{-i}$ must be at least as large as the original best-response advantage of the equilibrium strategy over $m_i$ against $s_{-i}$, as stated in \eqref{eq:full-cardinal-inequality}. Equation~\eqref{eq:full-difference} records the resulting transformed pay-off difference at each equilibrium.

\begin{remark}[Strict dominance of the entire Nash set]
\label{cor:strict-full-dominance-characterisation}
The strict version of \cref{thm:full-weak-dominance-characterisation} changes only the comparison in \eqref{eq:full-cardinal-inequality}: weak inequality is replaced by strict inequality for every $s\in\NE(G)$. Equivalently, because the non-negative pay-off additions may be multiplied by any positive constant, it is enough that, for each player $i$, the expected pay-off addition under $m_{-i}$ be strictly greater than under every equilibrium marginal $s_{-i}$. The security levels and any strictly positive dominance margins may then be prescribed arbitrarily.
\end{remark}
\noindent\emph{Proof.} See Appendix~\ref{app:strict-full-dominance-characterisation}.

\begin{corollary}[Strict Pareto dominance when all Nash equilibria are mixed]
\label{cor:all-mixed-full-dominance}
Suppose $|A_i|\geq2$ for both players, at least one player has at least three actions, and no Nash equilibrium strategy is pure. Then, for every prescribed pair of security levels $(\nu_1,\nu_2)\in\R^2$ and every pair of margins $\ell_1,\ell_2>0$, there exist a pure profile $a^M$ and a strategically equivalent game $\widetilde G$ such that
$a_i^M\in M_i(\widetilde G)$,
$v_i(\widetilde G)=\nu_i,$
and $\widetilde u_i(a^M)>\widetilde u_i(s)+\ell_i$ for every $s\in\NE(G)$ and $i\in N$.

\end{corollary}
\noindent\emph{Proof.} See Appendix~\ref{app:strict-full-dominance-characterisation}.

\subsection{Robustness theorems}
\label{subsec:robustness}

The preceding results establish that neither Pareto ranking is tied to a particular best-response structure, but they do not indicate whether the rankings persist under small changes in pay-offs. We consider this question next.

\begin{definition}[Dominance classes]
\label{def:dominance-classes-embedding}
Let $\mathcal G$ denote the class of all finite two-player normal-form games with real-valued pay-offs. Let $\mathcal C_M$ be the class of games in which some maximin profile strictly Pareto dominates every Nash equilibrium, and let $\mathcal C_N$ be the class of games in which some Nash equilibrium strictly Pareto dominates every maximin profile.
\end{definition}

Fix positive integers $m$ and $n$. After ordering the action sets, let $\mathcal G_{m,n}$ denote the space of $m\times n$ games, identified with $\R^{2mn}$ and equipped with its Euclidean topology. This identification lists the $2mn$ pure-profile pay-off numbers, so $G^k\to G$ means entrywise convergence of all those numbers. When dependence on the game matters, write $u_i^G$ and $M_i(G)$. Define the fixed-dimensional dominance classes by
\[
\mathcal C_M^{m,n}:=\mathcal C_M\cap\mathcal G_{m,n},
\qquad
\mathcal C_N^{m,n}:=\mathcal C_N\cap\mathcal G_{m,n}.
\]

Recall that membership in $\mathcal C_M$ requires the existence of one maximin profile that strictly Pareto dominates every Nash equilibrium, while membership in $\mathcal C_N$ requires the existence of one Nash equilibrium that strictly Pareto dominates every maximin profile. This quantifier structure is used in part \textup{(iii)} below.

\begin{definition}[Uniform dominance classes]
\label{def:uniform-dominance-classes}
For fixed $(m,n)$, let
\[
\mathcal U_M^{m,n}
:=
\left\{
	G\in\mathcal G_{m,n}:
	u_i^G(s^M)>u_i^G(s^N)
	~\forall i\in N,\ s^M\in M_1(G)\times M_2(G),\ s^N\in\NE(G)
	\right\},
\]
\[
\mathcal U_N^{m,n}
:=
\left\{
	G\in\mathcal G_{m,n}:
	u_i^G(s^N)>u_i^G(s^M)
	~\forall i\in N,\ s^N\in\NE(G),\ s^M\in M_1(G)\times M_2(G)
	\right\}.
\]
\end{definition}

The uniform classes require the relevant strict comparison to hold for every profile selected by the two solution concepts, rather than for one profile on the dominant side. Hence $\mathcal U_M^{m,n}\subseteq\mathcal C_M^{m,n}$ and $\mathcal U_N^{m,n}\subseteq\mathcal C_N^{m,n}$. For any $E\subseteq\mathcal G_{m,n}$, $\operatorname{int}(E)$ denotes the interior of $E$ relative to $\mathcal G_{m,n}$, that is, the set of games around which some open ball is contained in $E$. The symbol $\lambda_{2mn}$ denotes $2mn$-dimensional Lebesgue measure on $\mathcal G_{m,n}\cong\R^{2mn}$, and positive $\lambda_{2mn}$-measure means that a set occupies positive full-dimensional volume. To avoid confusion, any use below of ``generic'' or ``non-generic'' is measure-theoretic: it means, respectively, outside or within a measure-zero subset of the pay-off space.

\begin{theorem}[Main robustness theorem]
\label{thm:main-robustness}
For every pair of positive integers $(m,n)$:
\begin{enumerate}[label=\textup{(\roman*)}]
\item $\mathcal U_M^{m,n}$ and $\mathcal U_N^{m,n}$ are open in $\mathcal G_{m,n}$.
\item $\mathcal C_M^{m,n}$ and $\mathcal C_N^{m,n}$ are Lebesgue measurable and
\[
\lambda_{2mn}\!\left(\mathcal C_M^{m,n}\setminus\operatorname{int}(\mathcal C_M^{m,n})\right)=0,
\qquad
\lambda_{2mn}\!\left(\mathcal C_N^{m,n}\setminus\operatorname{int}(\mathcal C_N^{m,n})\right)=0.
\]
\item
$\lambda_{2mn}\!\left(\mathcal C_M^{m,n}\setminus\mathcal U_M^{m,n}\right)=0$, $ 
\mathcal U_M^{m,n}\subseteq\operatorname{int}(\mathcal C_M^{m,n})\subseteq\mathcal C_M^{m,n}$, and 
$\mathcal U_N^{m,n}\subseteq\operatorname{int}(\mathcal C_N^{m,n})\subseteq\mathcal C_N^{m,n}$.
\end{enumerate}
\end{theorem}
\noindent\emph{Proof.} See Appendix~\ref{app:robustness-theorems}.

Part \textup{(i)} says that the stronger uniform rankings survive all sufficiently small perturbations of the pay-off entries. Part \textup{(ii)} says that the larger existential classes may contain non-interior points, but those points occupy no full-dimensional volume. Part \textup{(iii)} gives a sharper conclusion on the maximin side: apart from a Lebesgue-null set, a game in $\mathcal C_M^{m,n}$ already satisfies the comparison for every maximin profile.

\begin{theorem}[Dimensional robustness theorem]
\label{thm:dimensional-robustness}
For every pair of positive integers $(m,n)$, the following are equivalent:
\begin{enumerate*}[label=\textup{(\roman*)}]
\item $\mathcal U_M^{m,n}\neq\varnothing$
\item $\operatorname{int}(\mathcal C_M^{m,n})\neq\varnothing$
\item $\lambda_{2mn}(\mathcal C_M^{m,n})>0$
\item $\min\{m,n\}\geq2$ and $\max\{m,n\}\geq3$.
\end{enumerate*}
\end{theorem}
\noindent\emph{Proof.} See Appendix~\ref{app:robustness-theorems}.

The dimensional result isolates the role of a third action. If one player has only one action, the maximin profile set and Nash equilibrium set coincide. In $2\times2$ games, \cref{lemma:2x2-maximin-no-strict-pareto} rules out strict maximin dominance of every equilibrium. Once both players have at least two actions and one has a third, the strict ranking occurs on a non-empty open region of the corresponding pay-off space. The third action therefore removes the one-sided low-dimensional advantage of Nash equilibrium; it does not create a general advantage for maximin, because the reverse ranking is also robust.

\begin{remark}[Nash dominance over maximin]
\label{rem:nash-robustness}
For every pair of positive integers $(m,n)$, the following are equivalent:
\begin{enumerate*}[label=\textup{(\roman*)}]
\item $\mathcal U_N^{m,n}\neq\varnothing$
\item $\operatorname{int}(\mathcal C_N^{m,n})\neq\varnothing$
\item $\lambda_{2mn}(\mathcal C_N^{m,n})>0$
\item $\min\{m,n\}\geq2$.
\end{enumerate*}
\end{remark}
\noindent\emph{Proof.} See Appendix~\ref{app:robustness-theorems}.

\begin{corollary}[Cardinality]
\label{cor:robustness-cardinality}
$
|\mathcal C_M|=|\mathcal C_N|=|\mathcal G|=2^{\aleph_0}.
$
\end{corollary}
\noindent\emph{Proof.} See Appendix~\ref{app:robustness-theorems}.

\cref{thm:maximin-nash-pareto,thm:maximin-pareto-characterisation,thm:full-weak-dominance-characterisation,thm:main-robustness,thm:dimensional-robustness} extend to $n$-player games, as discussed in the Appendix.
 
\section{Conclusion}
\label{sec:conclusion}

There is, we suspect, an influential informal evolutionary argument that non-maximising decision rules, leading players to adopt non-Nash equilibrium strategies, will prove damaging to those players' interests. We show instead that a maximin decision rule will not prove generally damaging to players as compared with what would have happened had they used Nash equilibrium strategies. 

This is important. Zero-sum thinking would seem to be well established in many rich countries and it might plausibly be associated with the wholesale use of a maximin decision rule. There are likely many reasons associated with the rise of populism that help explain this and which, in turn, will influence whether zero-sum thinking survives or continues to grow. Nevertheless, had that influential informal evolutionary argument proved correct, there would have been an underlying pay-off inferiority reason for supposing that zero-sum thinking in this sense of wholesale use of maximin might disappear, whatever its other attractions. We show, instead, that there is no pay-off reason of this kind for supposing that zero-sum thinking in this sense may prove transitory over the long run. Indeed, once established, it would seem from our results that this form of zero-sum thinking may be more difficult to dislodge or less likely to disappear than is, perhaps, often supposed.

While this is our main conclusion, our analysis also has implications for an older debate in game theory over what is the rational decision rule for players to use. Some notable game theorists have opined that maximin can seem as good a rule, if not better, than the maximising one that leads to the selection of Nash equilibrium strategies in some games. In relation to this debate, we have shown that the attractions of maximin are even stronger than previously recognised. For any Nash equilibrium of a game, one can always construct a strategically equivalent game in which the selected equilibrium retains its original pay-offs, every maximin profile weakly Pareto dominates it, and every maximin strategy guarantees its pay-off. Under mild conditions, a maximin profile can instead be made to Pareto dominate either that equilibrium or the entire Nash equilibrium set.  Moreover, strict dominance by a maximin profile is robust to small pay-off perturbations: it occurs on non-empty open sets of games whenever both players have at least two actions and one has at least three. In this context, it is perhaps also worth recalling that maximin has other, better known, reasons for being an attractive decision rule. One is informational: players only need to have knowledge of their own pay-offs with maximin whereas Nash requires knowledge (at least in a probabilistic sense) of the other player's pay-offs as well. Another is that maximin is not prone to multiplicity in the way that Nash equilibrium is. Since existence of multiple equilibria requires players to solve a coordination problem, this is a further apparent advantage of maximin. Indeed, when players fail to solve these coordination problems, the pay-offs to playing Nash equilibrium strategies decline and so the likelihood that maximin will deliver superior ones increases.

In short, the maximin rule has some significant advantages and what we show in this paper is that it has other advantages in terms of potential pay-off dominance that are not so well recognised.

\appendix

\section{Proofs for the main theorems}

\subsection{Preliminary lemmas}

\begin{lemma}[Crossed best responses]
\label{lem:crossed-best-responses}
Suppose $|A_i|\geq2$ for both players and at least one player has at least three actions. Then there are pure profiles $a$ and $b$ such that $a_i\in\BR_i^G(b_{-i})
\text{ and }
a_i\neq b_i
$ for all $i\in N$.
\end{lemma}

\begin{proof}
Select one pure best response $r_1(a_2)\in\BR_1(a_2)$ for each $a_2\in A_2$ and one pure best response $r_2(a_1)\in\BR_2(a_1)$ for each $a_1\in A_1$. Suppose no crossed pair existed. Then every $(b_1,b_2)\in A_1\times A_2$ would satisfy
$
b_1=r_1(b_2)
 \text{ or } 
b_2=r_2(b_1).
$
Hence $A_1\times A_2$ would be covered by the two graphs
$
\{(r_1(b_2),b_2):b_2\in A_2\}
 \text{ and }
\{(b_1,r_2(b_1)):b_1\in A_1\},
$
which together contain at most $|A_1|+|A_2|$ profiles. But
$|A_1||A_2|>|A_1|+|A_2|$
when both action sets have at least two elements and one has at least three. This contradiction gives a profile $b$ outside both graphs. Set $a_i=r_i(b_{-i})$. Then $a_i\in\BR_i^G(b_{-i})$ and $a_i\neq b_i$ for both players. 
\end{proof}

\begin{lemma}[Compactness of the Nash set]
\label{lem:nash-compact}
For every finite game, $\NE(G)$ is nonempty and compact.
\end{lemma}

\begin{proof}
Nonemptiness is Nash's existence theorem \citep{nash1951}. The product $S_1\times S_2$ is compact. For each player $i$ and each pure action $a_i$, the condition
$
u_i(s)\geq u_i(a_i,s_{-i})
$
defines a closed subset of $S_1\times S_2$ because expected pay-offs are continuous. The Nash set is the intersection of these finitely many closed sets, so it is closed and therefore compact.
\end{proof}

\begin{lemma}[Target-pay-off implementation]
\label{lem:target-payoff-implementation}
Fix $m_i\in S_i$ and a target pay-off function $w_i\colon A_{-i}\to\R$. Define
$
h_i(a_{-i}):=w_i(a_{-i})-u_i(m_i,a_{-i})
$
and let $\widetilde u_i=u_i+h_i$. Then, for every $\sigma_i\in S_i$ and $q_{-i}\in S_{-i}$,
\begin{equation}
\widetilde u_i(\sigma_i,q_{-i})
=
\sum_{a_{-i}\in A_{-i}}q_{-i}(a_{-i})w_i(a_{-i})
+u_i(\sigma_i,q_{-i})-u_i(m_i,q_{-i}).
\label{eq:target-payoff-identity}
\end{equation}
The expected target pay-off under $q_{-i}$ is given by $\widetilde u_i(m_i,q_{-i})
=
\sum_{a_{-i}\in A_{-i}}q_{-i}(a_{-i})w_i(a_{-i})$.
\end{lemma}

\begin{proof}
By construction, \eqref{eq:target-payoff-identity} follows from
\begin{align*}
\widetilde u_i(\sigma_i,q_{-i})
&=u_i(\sigma_i,q_{-i})
+\sum_{a_{-i}\in A_{-i}}q_{-i}(a_{-i})h_i(a_{-i})\\
&=u_i(\sigma_i,q_{-i})
+\sum_{a_{-i}\in A_{-i}}q_{-i}(a_{-i})w_i(a_{-i})
-u_i(m_i,q_{-i}).
\end{align*}
\end{proof}

\begin{lemma}[Maximin implementation]
\label{lem:maximin-implementation}
Fix $m_i\in S_i$. Suppose there are $\mu_{-i}\in S_{-i}$, $\nu_i\in\R$, and non-negative pay-off additions $r_i(a_{-i})\geq0$, one for each $a_{-i}\in A_{-i}$, such that
$
m_i\in\BR_i^G(\mu_{-i}),
$
$\sum_{a_{-i}\in A_{-i}}\mu_{-i}(a_{-i})r_i(a_{-i})=0.
$
For any $K_i>0$, implement the target pay-off function
$
w_i(a_{-i})=\nu_i+K_i r_i(a_{-i}).
$
Then
$
m_i\in M_i(\widetilde G)$,
$v_i(\widetilde G)=\nu_i.
$
If $m_i=a_i$ is pure and is the unique best response to $\mu_{-i}$, then
$
M_i(\widetilde G)=\{a_i\}.
$
\end{lemma}

\begin{proof}
Because every pay-off addition is non-negative, $m_i$ obtains at least $\nu_i$ against every pure opponent action. For any $\sigma_i\in S_i$,
$
\min_{a_{-i}\in A_{-i}}
\widetilde u_i(\sigma_i,a_{-i})
\leq
\widetilde u_i(\sigma_i,\mu_{-i}).
$
By \cref{lem:target-payoff-implementation}, the definition of
$w_i$, the condition $\sum_{a_{-i}}\mu_{-i}(a_{-i})r_i(a_{-i})=0$,
and $m_i\in\BR_i^G(\mu_{-i})$,
\begin{align*}
\widetilde u_i(\sigma_i,\mu_{-i})
&=
\sum_{a_{-i}\in A_{-i}}
\mu_{-i}(a_{-i})w_i(a_{-i})
+
u_i(\sigma_i,\mu_{-i})
-
u_i(m_i,\mu_{-i})\\
&=
\nu_i
+
K_i\sum_{a_{-i}\in A_{-i}}
\mu_{-i}(a_{-i})r_i(a_{-i})
+
u_i(\sigma_i,\mu_{-i})
-
u_i(m_i,\mu_{-i}) \leq \nu_i.
\end{align*}

Since \(m_i\) guarantees at least \(\nu_i\), while every strategy
\(\sigma_i\in S_i\) guarantees at most \(\nu_i\), it follows that
$v_i(\widetilde G)=\nu_i$ and $
m_i\in M_i(\widetilde G)$.

If \(m_i=a_i\) is pure and is the unique best response to
\(\mu_{-i}\), then every \(\sigma_i\neq a_i\) satisfies
$
\min_{a_{-i}\in A_{-i}}
\widetilde u_i(\sigma_i,a_{-i})<\nu_i.
$
Hence no other strategy is maximin, and $
M_i(\widetilde G)=\{a_i\}$.
\end{proof}

\subsection{Proofs for comparison with a selected Nash equilibrium}
\label{app:selected-equilibrium}

\begin{proof}[Proof of \cref{thm:maximin-nash-pareto}]
For each player \(i\), let $U_i:=u_i(\bar s)$. Since $G$ is finite, there is a pure best response $b_i\in \BR_i^G(\bar s_{-i})$ such that
\begin{equation}
u_i(b_i,\bar s_{-i})
=
u_i(\bar s_i,\bar s_{-i})
=
U_i.
\label{eq:selected-security-best-response}
\end{equation}

Choose the constant target pay-off function $w_i(a_{-i}):=U_i$ for every $a_{-i}\in A_{-i}$. By \cref{lem:target-payoff-implementation}, this target is implemented by
$
h_i(a_{-i})
=
w_i(a_{-i})-u_i(b_i,a_{-i})
=
U_i-u_i(b_i,a_{-i}),
$
so the transformed pay-off function is
$
\widehat u_i(a_i,a_{-i})
=
u_i(a_i,a_{-i})+h_i(a_{-i}).
$
The adjustment depends only on the opponent's action. Hence
\(\widehat G\) is strategically equivalent to \(G\), and $\NE(\widehat G)=\NE(G)$.

Apply \cref{lem:maximin-implementation} with
$
m_i=b_i,
\mu_{-i}=\bar s_{-i},
\nu_i=U_i,
r_i(a_{-i})=0$
for every $a_{-i}$.

The required conditions hold because $b_i\in\BR_i^G(\bar s_{-i})$ and $
\sum_{a_{-i}} \bar s_{-i}(a_{-i})r_i(a_{-i})=0$.
Therefore
\begin{equation}
b_i\in M_i(\widehat G)
\quad\text{and}\quad
v_i(\widehat G)=U_i.
\label{eq:selected-security-equality}
\end{equation}

The selected equilibrium pay-off is unchanged. Indeed, by
\cref{lem:target-payoff-implementation} and
\eqref{eq:selected-security-best-response},
\begin{align*}
\widehat u_i(\bar s)
&=
\sum_{a_{-i}}
\bar s_{-i}(a_{-i})w_i(a_{-i})
+
u_i(\bar s)
-
u_i(b_i,\bar s_{-i}) =
U_i+U_i-U_i =U_i.
\end{align*}
Together with \eqref{eq:selected-security-equality}, this gives
$
\widehat u_i(\bar s)
=
v_i(\widehat G)
=
u_i(\bar s)
$ for all $i\in N$, 
proving part \textup{(i)}.

Now let $m$ be any mutual-maximin profile. Since \(m_i\) guarantees the security level $v_i(\widehat G)=U_i$,
$
\widehat u_i(m)
\geq
U_i
=
\widehat u_i(\bar s)
$ for all $i\in N$.
Thus every mutual-maximin profile weakly Pareto dominates the selected
equilibrium.

Finally, let $a^M:=(b_1,b_2)$.
By \eqref{eq:selected-security-equality}, \(a^M\) is a pure
mutual-maximin profile. Since the target pay-off of \(b_i\) is constant,
$
\widehat u_i(a^M)
=
\widehat u_i(b_i,b_{-i})
=
w_i(b_{-i})
=
U_i
=
u_i(\bar s)
$ for all $i\in N$. 
This proves part \textup{(iii)}.
\end{proof}

\begin{proof}[Proof of \cref{thm:maximin-pareto-characterisation}]
Suppose first that a strategically equivalent game $\widetilde G$ makes $m$ a maximin profile and satisfies
$
\widetilde u_i(m)>\widetilde u_i(\bar s).
$
Let $v_i$ denote player $i$'s transformed security level and define
$
r_i(a_{-i}):=\widetilde u_i(m_i,a_{-i})-v_i.
$
Since $m_i$ is maximin, it guarantees $v_i$, so $r_i(a_{-i})\geq 0$ for every $a_{-i}\in A_{-i}$.

Finite minimax duality in the zero-sum problem with pay-off $\widetilde u_i$ supplies a minimax belief $\mu_{-i}$ such that
$
m_i\in\BR_i^{\widetilde G}(\mu_{-i})$,
$\widetilde u_i(m_i,\mu_{-i})=v_i.
$
It follows that
$
\sum_{a_{-i}\in A_{-i}}\mu_{-i}(a_{-i})r_i(a_{-i})=0.
$
Strategic equivalence preserves best responses, so $m_i\in\BR_i^G(\mu_{-i})$. This proves \eqref{eq:selected-support-addition}.

Since $\widetilde u_i(m_i,a_{-i})=v_i+r_i(a_{-i})$, \cref{lem:target-payoff-implementation} gives
\begin{equation}
\widetilde u_i(m)-\widetilde u_i(\bar s)
=
\sum_{a_{-i}\in A_{-i}}
\bigl(m_{-i}(a_{-i})-\bar s_{-i}(a_{-i})\bigr)r_i(a_{-i})
-\bigl[ u_i(\bar s)-u_i(m_i, \bar s_{-i}) \bigr].
\label{eq:selected-difference-proof}
\end{equation}
The left-hand side is positive and $u_i(\bar s)-u_i(m_i, \bar s_{-i})\geq0$. Hence the expected pay-off addition under $m_{-i}$ is strictly greater than the expected pay-off addition under $\bar s_{-i}$, which proves \eqref{eq:selected-separation} and therefore necessity.

Conversely, suppose the conditions in \cref{thm:maximin-pareto-characterisation}\textup{(b)} hold, and write
\[
d_i:=
\sum_{a_{-i}\in A_{-i}}
\bigl(m_{-i}(a_{-i})-\bar s_{-i}(a_{-i})\bigr)r_i(a_{-i})>0.
\]
Given $\ell_i>0$, choose $K_i>0$ so that
$
K_i d_i>u_i(\bar s)-u_i(m_i, \bar s_{-i})+\ell_i.
$
Set
\[
\nu_i:=u_i(m_i,\bar s_{-i})
-K_i\sum_{a_{-i}\in A_{-i}}\bar s_{-i}(a_{-i})r_i(a_{-i})
\]
and implement
$
w_i(a_{-i})=\nu_i+K_i r_i(a_{-i}).
$

By \cref{lem:maximin-implementation}, the conditions in \eqref{eq:selected-support-addition} imply that
$
m_i\in M_i(\widetilde G)
$ for all $i \in N$.

We next verify that the selected equilibrium pay-off is preserved. By \eqref{eq:target-payoff-identity},
\begin{align*}
\widetilde u_i(\bar s)
&=
\sum_{a_{-i}\in A_{-i}}
\bar s_{-i}(a_{-i})w_i(a_{-i})
+
u_i(\bar s)
-
u_i(m_i,\bar s_{-i})\\
&=
\nu_i
+
K_i\sum_{a_{-i}\in A_{-i}}
\bar s_{-i}(a_{-i})r_i(a_{-i})
+
u_i(\bar s)
-
u_i(m_i,\bar s_{-i}).
\end{align*}
Substituting the definition
$
\nu_i
=
u_i(m_i,\bar s_{-i})
-
K_i\sum_{a_{-i}\in A_{-i}}
\bar s_{-i}(a_{-i})r_i(a_{-i})
$
gives
$
\widetilde u_i(\bar s)=u_i(\bar s).
$

For the candidate profile, \eqref{eq:target-payoff-identity} gives
\begin{align*}
\widetilde u_i(m)
&=
\sum_{a_{-i}\in A_{-i}}
m_{-i}(a_{-i})w_i(a_{-i}) =
\nu_i
+
K_i\sum_{a_{-i}\in A_{-i}}
m_{-i}(a_{-i})r_i(a_{-i}).
\end{align*}
Hence
\begin{align*}
\widetilde u_i(m)-u_i(\bar s)
&=
K_i
\sum_{a_{-i}\in A_{-i}}
\bigl(m_{-i}(a_{-i})-\bar s_{-i}(a_{-i})\bigr)
r_i(a_{-i}) -\bigl[u_i(\bar s)-u_i(m_i,\bar s_{-i})\bigr]\\
&=
K_i d_i
-
\bigl[u_i(\bar s)-u_i(m_i,\bar s_{-i})\bigr] >
\ell_i,
\end{align*}
where the last inequality follows from the choice of \(K_i\). Therefore,
$
\widetilde u_i(m)
>
u_i(\bar s)+\ell_i
=
\widetilde u_i(\bar s)+\ell_i.
$
Thus the selected equilibrium retains its original pay-off, while \(m\) strictly Pareto dominates it by the prescribed margins.
\end{proof}

\begin{proof}[Proof of \cref{cor:pure-selected-separation}]
For each player \(i\), let the supporting belief assign probability one to \(b_{-i}\), and set the pay-off addition equal to one for \(a_{-i}^M\) and zero to every other opponent action. The supporting belief has expected pay-off addition zero, while
\[
\sum_{a_{-i}\in A_{-i}}
\bigl(a_{-i}^M(a_{-i})-\bar s_{-i}(a_{-i})\bigr)r_i(a_{-i})
=
1-\bar s_{-i}(a_{-i}^M)>0.
\]
The result follows from \cref{thm:maximin-pareto-characterisation}. Strictness of the supporting best response gives uniqueness of the constructed maximin strategy.
\end{proof}

\begin{proof}[Proof of \cref{cor:mixed-selected-separation}]

The crossed-best-response lemma provides pure profiles \(a^M=(a_1^M,a_2^M)\) and \(b=(b_1,b_2)\) such that
$
a_i^M\in\BR_i^G(b_{-i})$ and 
$a_i^M\neq b_i
$ for all $i \in N$.
Because each equilibrium strategy has support of size at least two,
$
\bar s_i(a_i^M)<1
$ for all $i \in N$.
The preceding corollary therefore applies.
\end{proof}

\subsection{Proof of the full weak-dominance characterisation}
\label{app:full-weak-dominance-characterisation}

\begin{proof}[Proof of \cref{thm:full-weak-dominance-characterisation}]
	Suppose first that the conditions in \eqref{eq:full-support-addition} and \eqref{eq:full-cardinal-inequality} hold. For prescribed security levels $\nu_i$, implement $w_i(a_{-i})=\nu_i+r_i(a_{-i}).$ By \cref{lem:maximin-implementation}, $m_i$ is maximin and $v_i(\widetilde G)=\nu_i$. For any $s\in\NE(G)$, \eqref{eq:target-payoff-identity} gives
	\[
	\widetilde u_i(s)
	=\nu_i+
	\sum_{a_{-i}\in A_{-i}}s_{-i}(a_{-i})r_i(a_{-i})
	+\bigl[u_i(s)-u_i(m_i,s_{-i})\bigr],
	\]
	\[
	\widetilde u_i(m)
	=\nu_i+
	\sum_{a_{-i}\in A_{-i}}m_{-i}(a_{-i})r_i(a_{-i}).
	\]
	Therefore, the following proves sufficiency and \eqref{eq:full-difference}:
    
    \noindent 
	$\widetilde u_i(m)-\widetilde u_i(s) = $
		$\sum_{a_{-i}\in A_{-i}}
		\bigl(m_{-i}(a_{-i})-s_{-i}(a_{-i})\bigr)r_i(a_{-i})
		-\bigl[u_i(s)-u_i(m_i,s_{-i})\bigr] \geq0.$

	For necessity, suppose a strategically equivalent game $\widetilde G$ makes $m$ a maximin profile and weakly Pareto dominant over every equilibrium. Let $v_i$ be the transformed security level and define $r_i(a_{-i}):=\widetilde u_i(m_i,a_{-i})-v_i\geq0.$ Finite minimax duality supplies a belief $\mu_{-i}$ such that
	$m_i\in\BR_i^{\widetilde G}(\mu_{-i})$ and $\widetilde u_i(m_i,\mu_{-i})=v_i.$
	Hence $\sum_{a_{-i}\in A_{-i}}\mu_{-i}(a_{-i})r_i(a_{-i})=0.$ Strategic equivalence transfers the best-response condition to $G$. As in \eqref{eq:selected-difference-proof}, for every $s\in\NE(G)$,
	\[
	\widetilde u_i(m)-\widetilde u_i(s)
	=
	\sum_{a_{-i}\in A_{-i}}
	\bigl(m_{-i}(a_{-i})-s_{-i}(a_{-i})\bigr)r_i(a_{-i})
	-\bigl[u_i(s)-u_i(m_i,s_{-i})\bigr].
	\]
	Weak dominance therefore gives \eqref{eq:full-cardinal-inequality}.
\end{proof}

\subsection{Strict dominance of the entire Nash set}
\label{app:strict-full-dominance-characterisation}

\begin{proof}[Proof of \cref{cor:strict-full-dominance-characterisation}]
	Necessity follows from the necessity argument in \cref{thm:full-weak-dominance-characterisation}. If $\widetilde u_i(m)>\widetilde u_i(s)$ for every equilibrium, then $\sum_{a_{-i}\in A_{-i}} \bigl(m_{-i}(a_{-i})-s_{-i}(a_{-i})\bigr)r_i(a_{-i}) >u_i(s)-u_i(m_i,s_{-i})\geq0$ for every $s\in\NE(G)$. In particular, the expected pay-off addition under $m_{-i}$ is strictly greater than the expected pay-off addition under every equilibrium marginal.
	
	Conversely, suppose the best-response and zero-expected-addition conditions in \eqref{eq:full-support-addition} hold and, for every $s\in\NE(G)$, $d_i(s):= \sum_{a_{-i}\in A_{-i}} \bigl(m_{-i}(a_{-i})-s_{-i}(a_{-i})\bigr)r_i(a_{-i})>0.$ It is well known that $\NE(G)$ is compact (for completeness, see \cref{lem:nash-compact}). Since $d_i$ is continuous, $\delta_i:=\min_{s\in\NE(G)}d_i(s)>0.$ The function $s\longmapsto u_i(s)-u_i(m_i,s_{-i})$ is continuous and therefore bounded on the compact Nash set. Let $B_i:=\max_{s\in\NE(G)} \bigl[u_i(s)-u_i(m_i,s_{-i})\bigr].$ Given a prescribed margin $\ell_i>0$, choose $K_i>\frac{B_i+\ell_i}{\delta_i}$ and implement $w_i(a_{-i})=\nu_i+K_i r_i(a_{-i}).$ By \cref{lem:maximin-implementation}, the security level is $\nu_i$. For every equilibrium $s$, $\widetilde u_i(m)-\widetilde u_i(s) =K_i d_i(s)-\bigl[u_i(s)-u_i(m_i,s_{-i})\bigr] \geq K_i\delta_i-B_i >\ell_i.$ Thus the security levels and strictly positive dominance margins may be prescribed arbitrarily.
\end{proof}

\begin{proof}[Proof of \cref{cor:all-mixed-full-dominance}]
	By \cref{lem:crossed-best-responses}, choose pure profiles $a^M=(a_1^M,a_2^M)$ and $b=(b_1,b_2)$ such that
	$a_i^M\in\BR_i^G(b_{-i})$ and $a_i^M\neq b_i$ for all $i \in N$.
	For each player, take $\mu_{-i}=\delta_{b_{-i}}$ and set the pay-off addition equal to one for $a_{-i}^M$ and equal to zero for every other opponent action. The expected pay-off addition under $\mu_{-i}$ is zero. For every equilibrium $s$, $\sum_{c_{-i}\in A_{-i}} \bigl(a_{-i}^M(c_{-i})-s_{-i}(c_{-i})\bigr)r_i(c_{-i}) =1-s_{-i}(a_{-i}^M)>0,$ because $s_{-i}$ is not pure. The conclusion follows from \cref{cor:strict-full-dominance-characterisation}.
\end{proof}

\subsection{Proofs for the robustness theorems}
\label{app:robustness-theorems}

\begin{lemma}
	\label{lem:rectangular-padding}
	Let $G\in\mathcal G_{m,n}$ and let $m'\geq m$ and $n'\geq n$. There exists a game $\widehat G\in\mathcal G_{m',n'}$ such that, under the natural embedding of the old mixed-strategy spaces, $\NE(\widehat G)=\NE(G)$,  $ M_i(\widehat G)=M_i(G)$ for all $i\in N$, and every strategy profile supported on the old actions has the same pay-off as in $G$. The natural embedding identifies an old mixed strategy with the enlarged strategy that assigns zero probability to every added action. Consequently, membership in either uniform dominance class is preserved.
\end{lemma}

\begin{proof}
	Let $A_1^+$ and $A_2^+$ denote the added row and column actions. Choose constants
	$\alpha_i<\min_{a\in A_1\times A_2}u_i(a)$, $\beta_i>\max_{a\in A_1\times A_2}u_i(a)$
	for all $i\in N$.
	Retain the old pay-off block. At an old-row/new-column cell set the pay-off vector equal to $(\beta_1,\alpha_2)$; at a new-row/old-column cell set it equal to $(\alpha_1,\beta_2)$; and at a new-row/new-column cell set it equal to $(\alpha_1,\alpha_2)$. Every new action is strictly dominated by every old action of the same player. Hence no Nash equilibrium uses a new action, and all best-response comparisons within the old block are unchanged. It follows that $\NE(\widehat G)=\NE(G)$ under the embedding just described.
	
	For an old mixed strategy of player $i$, every new opponent action gives the high pay-off $\beta_i$, so its worst-case pay-off remains exactly its old worst-case pay-off. Let $v_i$ be player $i$'s security level in $G$, and suppose an own strategy assigns total probability $\rho>0$ to new actions. If $\rho<1$, write the strategy as $(1-\rho)y_i+\rho d_i$, where $y_i$ uses only old actions and $d_i$ uses only new actions. Choose an old opponent action that minimises the old pay-off of $y_i$. Against that action, the resulting pay-off is at most $(1-\rho)v_i+\rho\alpha_i<v_i.$ If $\rho=1$, every old opponent action gives $\alpha_i<v_i$. Thus no strategy using a new own action is maximin, while every old strategy retains its former guarantee. Therefore $M_i(\widehat G)=M_i(G)$ for each player. Since the relevant profile sets and their pay-offs are unchanged, either uniform dominance ranking is preserved.
\end{proof}

\begin{proof}[Proof of \cref{thm:main-robustness}]
	Fix $(m,n)$.  For part \textup{(i)}, define player $i$'s security function by $\phi_i^G(s_i):=\min_{a_{-i}\in A_{-i}}u_i^G(s_i,a_{-i}).$ It is continuous in $(G,s_i)$ because it is the minimum of finitely many jointly continuous functions. Since $S_i$ is compact, this function attains its maximum, so $M_i(G)$ is non-empty and compact. The maximin correspondence has a closed graph: if $G^k\to G$, $s_i^k\to s_i$, and $s_i^k\in M_i(G^k)$, then, for every $t_i\in S_i$, $\phi_i^{G^k}(s_i^k)\geq\phi_i^{G^k}(t_i),$ and passage to the limit gives $\phi_i^G(s_i)\geq\phi_i^G(t_i)$, so $s_i\in M_i(G)$.\footnote{A correspondence $F$ assigns to each game $G$ a set $F(G)$ of strategies. Its graph is closed if, whenever $G^k\to G$, $x^k\to x$, and $x^k\in F(G^k)$, then $x\in F(G)$. Here $x^k$ is an element selected from the set $F(G^k)$.}
	
	The Nash correspondence also has non-empty compact values. Non-emptiness is Nash's existence theorem \citep{nash1951}; for a fixed game, the Nash inequalities define a closed subset of the compact space $S_1\times S_2$, so $\NE(G)$ is compact. Its graph is closed as well: if $G^k\to G$, $s^k\to s$, and $s^k\in\NE(G^k)$, then, for every player $i$ and every pure action $a_i$, $u_i^{G^k}(s^k)\geq u_i^{G^k}(a_i,s_{-i}^k).$ Taking limits gives the Nash inequalities in $G$, so $s\in\NE(G)$.

	Define the auxiliary functions
	\[
	\Delta_M(G):=
	\min_{\substack{i\in N,\ s^M\in M_1(G)\times M_2(G),\ s^N\in\NE(G)}}
	\left[u_i^G(s^M)-u_i^G(s^N)\right]
	\]
	\[
	\Delta_N(G):=
	\min_{\substack{i\in N,\ s^N\in\NE(G),\ s^M\in M_1(G)\times M_2(G)}}
	\left[u_i^G(s^N)-u_i^G(s^M)\right].
	\]
	The minima are attained because both correspondences have non-empty compact values. Passing to a subsequence that converges to the $\liminf$ isolates the smallest limiting values, so proving the required inequality along that subsequence establishes lower semicontinuity.\footnote{A function $\Delta$ is lower semicontinuous if $\Delta(G)\leq\liminf_{k\to\infty}\Delta(G^k)$ whenever $G^k\to G$. To verify this inequality, one may pass to a subsequence along which $\Delta(G^k)$ converges to the $\liminf$ of the original sequence.}
	
	We prove the claim for $\Delta_M$; the proof for $\Delta_N$ is identical. Let $G^k\to G$ and pass to a subsequence along which $\Delta_M(G^k)$ converges to the lower limit of the original sequence. For each $k$, choose a player $i_k$, a maximin profile $s^{M,k}$, and a Nash equilibrium $s^{N,k}$ attaining the minimum. Compactness permits a further subsequence such that $i_k=i$, $s^{M,k}\to s^M$, and $s^{N,k}\to s^N$. The closed-graph arguments give $s^M\in M_1(G)\times M_2(G)$ and $s^N\in\NE(G)$. Hence
	\[
	\liminf_{k\to\infty}\Delta_M(G^k)
	=
	u_i^G(s^M)-u_i^G(s^N)
	\geq\Delta_M(G).
	\]
	Thus $\Delta_M$ is lower semicontinuous, and the same holds for $\Delta_N$. For a lower semicontinuous function, every strict superlevel set is open: if $\Delta(G)>0$, lower semicontinuity gives a neighbourhood on which $\Delta>\Delta(G)/2>0$. Since $\mathcal U_M^{m,n}=\{G:\Delta_M(G)>0\}$ and $ \mathcal U_N^{m,n}=\{G:\Delta_N(G)>0\},$ both uniform classes are open.

	For part \textup{(ii)}, recall that a semialgebraic set is a subset of a Euclidean space that can be described by finitely many polynomial equalities and inequalities, combined using finitely many logical operations such as ``and'', ``or'', and ``not''. Simplex membership, the Nash equilibrium conditions, maximin optimality, and the relevant strict pay-off comparisons can all be written as first-order formulas over the real numbers involving such polynomial conditions. Here, a first-order formula may also quantify over finitely many real-valued auxiliary variables, such as mixed-strategy probabilities and security levels.
	
	For example, $x_i\in M_i(G)$ if and only if $x_i\in S_i$ and there exists $z_i\in\R$ such that $u_i^G(x_i,a_{-i})\geq z_i\text{ for every }a_{-i}\in A_{-i},$ while, for every $y_i\in S_i$, there exists $a_{-i}\in A_{-i}$ such that $u_i^G(y_i,a_{-i})\leq z_i.$ The first set of inequalities says that $x_i$ guarantees at least $z_i$, while the second says that no alternative strategy can guarantee more than $z_i$. Hence $x_i$ is maximin. Because the action sets are finite, the phrases ``for every $a_{-i}\in A_{-i}$'' and ``there exists $a_{-i}\in A_{-i}$'' amount only to finite conjunctions and disjunctions of polynomial inequalities.
	
	The definitions of $\mathcal C_M^{m,n}$ and $\mathcal C_N^{m,n}$ can therefore be expressed by first-order formulas involving the pay-off entries of the game together with finitely many auxiliary variables representing maximin strategies, Nash equilibria, and security levels. The Tarski-Seidenberg theorem states that the projection\footnote{That is, after describing all games and auxiliary variables satisfying the relevant conditions, one may discard the auxiliary variables and retain only the corresponding pay-off vectors of the games.} of a semialgebraic set is semialgebraic. The resulting set of games remains semialgebraic. Since semialgebraic sets are also closed under finite unions, intersections, and complements, both $\mathcal C_M^{m,n}$ and $\mathcal C_N^{m,n}$ are semialgebraic. In particular, they are Lebesgue measurable.
	
	We finally use a standard geometric property of semialgebraic sets. Every semialgebraic subset of $\R^{2mn}$ can be partitioned into finitely many smooth pieces, called cells, each of some dimension at most $2mn$. A cell of dimension $2mn$ is locally open in $\R^{2mn}$ and therefore consists of interior points of the set. Consequently, the portion of a semialgebraic set lying outside its interior can contain only cells of dimension strictly below $2mn$. Such lower-dimensional pieces have zero $2mn$-dimensional Lebesgue measure. Therefore
	$
	\lambda_{2mn}\!\left(\mathcal C_M^{m,n}\setminus\operatorname{int}(\mathcal C_M^{m,n})\right)=0,
	$
	and
	$
	\lambda_{2mn}\!\left(\mathcal C_N^{m,n}\setminus\operatorname{int}(\mathcal C_N^{m,n})\right)=0.
	$
	This proves part \textup{(ii)}.\footnote{See Chapter~2 of \citet{bochnak2013} for the Tarski-Seidenberg theorem and the relevant properties of semialgebraic sets.}

	For part \textup{(iii)}, fix a player $i$ and let $\mathcal D_i^{m,n} := \left\{ G\in\mathcal G_{m,n}:M_i(G)\text{ is not a singleton} \right\}.$ This set is semialgebraic because it is described by the existence of two distinct points in the semialgebraic graph of the maximin correspondence.

	Player $i$'s maximin strategies in $G$ are precisely her optimal strategies in the auxiliary two-player zero-sum game in which her pay-off function is $u_i^G$ and the opponent's pay-off function is $-u_i^G$. Let $U_{m,n}^{\mathrm{zs}}$ denote the set of $m\times n$ pay-off matrices for which both players in the associated zero-sum game have unique optimal strategies. \citet[Section~6, Theorem~3]{bohnenblust1950} prove that $U_{m,n}^{\mathrm{zs}}$ is open and everywhere dense in $\R^{mn}$.
	
	In particular, the set of matrices for which the relevant player has a unique optimal strategy contains $U_{m,n}^{\mathrm{zs}}$ and is therefore dense, although it need not be open. Consequently, for every game $G$ and every neighbourhood of $G$ in $\mathcal G_{m,n}$, one can perturb only player $i$'s pay-offs by an arbitrarily small amount so that, in the resulting auxiliary zero-sum game, player $i$ has a unique optimal strategy. For player $2$, this amounts to applying the zero-sum result to the transpose of her pay-off matrix.
	
	Thus every neighbourhood of every game contains a game outside $\mathcal D_i^{m,n}$, so $\mathcal D_i^{m,n}$ has empty interior. Since it is semialgebraic, the dimension result used in part \textup{(ii)} implies that it has zero $2mn$-dimensional Lebesgue measure: $\lambda_{2mn}\!\left( \mathcal D_i^{m,n} \right) = 0.$ Therefore, outside a measure-zero subset of $\mathcal G_{m,n}$, player $i$ has a unique maximin strategy.
	
	Now suppose $G \in \mathcal C_M^{m,n} \setminus \mathcal U_M^{m,n}.$ By the definition of $\mathcal C_M^{m,n}$, some maximin profile strictly Pareto dominates every Nash equilibrium. If both players had unique maximin strategies, there would be only one maximin profile. That profile would then strictly Pareto dominate every Nash equilibrium, implying $G\in\mathcal U_M^{m,n},$ contrary to the assumption. Hence at least one player's maximin set is not a singleton, and
	$
	\mathcal C_M^{m,n}
	\setminus
	\mathcal U_M^{m,n}
	\subseteq
	\mathcal D_1^{m,n}
	\cup
	\mathcal D_2^{m,n}.
	$
	The right-hand side is a finite union of measure-zero sets. Consequently,
	$
	\lambda_{2mn}\!\left(
	\mathcal C_M^{m,n}
	\setminus
	\mathcal U_M^{m,n}
	\right)
	=
	0.
	$
	
	Finally, part \textup{(i)} shows that $\mathcal U_M^{m,n}$ and $\mathcal U_N^{m,n}$ are open. Since
	$
	\mathcal U_M^{m,n}
	\subseteq
	\mathcal C_M^{m,n}$
	 and 
	$\mathcal U_N^{m,n}
	\subseteq
	\mathcal C_N^{m,n},
	$
	it follows that
	$
	\mathcal U_M^{m,n}
	\subseteq
	\operatorname{int}(\mathcal C_M^{m,n})
	\subseteq
	\mathcal C_M^{m,n}
	$
	and
	$
	\mathcal U_N^{m,n}
	\subseteq
	\operatorname{int}(\mathcal C_N^{m,n})
	\subseteq
	\mathcal C_N^{m,n}.
	$
	This proves part \textup{(iii)}.
\end{proof}

\begin{lemma}\label{lemma:2x2-maximin-no-strict-pareto}
Let $G$ be a $2\times 2$ game. No maximin profile in $G$ can strictly Pareto dominate all Nash equilibria.
\end{lemma}

\begin{remark}
Lemma~\ref{lemma:2x2-maximin-no-strict-pareto} does not extend to
$2\times3$ games. Consider
\[
\begin{array}{c@{~~} | @{~~}c@{~~}c@{~~}c}
 & L & M & R\\ \hline
T & 1,1 & 0,0 & 2,0\\
B & 3,3 & 1,3 & 0,4
\end{array}
\]
Player $1$'s unique maximin strategy is $(\frac13, \frac23)$, player $2$'s unique maximin strategy is $L$, and the resulting maximin profile yields pay-off $(\frac73,\frac73).$
The game has a unique Nash equilibrium,
$\left( (\frac12, \frac12),\ (\frac12, 0, \frac12)\right),
$ with pay-off $(\frac32,2).$
Hence the maximin profile strictly Pareto dominates the unique Nash equilibrium.
\end{remark}

\begin{proof}[Proof of \cref{lemma:2x2-maximin-no-strict-pareto}]
Write the game as
\[
\begin{array}{c@{~~} | @{~~}c@{~~}c}
 & L & R\\ \hline
T & a,e & b,f\\
B & c,g & d,h
\end{array}
\]
and let $p$ be player $1$'s probability of $T$ and $q$ player $2$'s
probability of $L$.

Because the game is $2\times2$, for any fixed $p$,
$\min_{q'\in[0,1]} u_1(p,q')
=
\min\{u_1(p,L),u_1(p,R)\}$,
since $u_1(p,q')$ is linear in $q'$. And, for any  $q$,
$\min_{p'\in[0,1]} u_2(p',q)=\min\{u_2(T,q),u_2(B,q)\}.$ Since $p$ and $q$ are maximin, $v_1=\min\{u_1(p,L),u_1(p,R)\}$ and $v_2=\min\{u_2(T,q),u_2(B,q)\}.$

Assume, toward contradiction, that $u(p,q)$ strictly Pareto dominates every Nash equilibrium pay-off profile. Let $(p^*,q^*)\in \NE(G)$. Since $p^*$ is a best response to $q^*$, 
$u_1(p^*,q^*)\ge u_1(p,q^*)\ge \min\{u_1(p,L),u_1(p,R)\}=v_1.$
Similarly, since $q^*$ is a best response to $p^*$,
$u_2(p^*,q^*)\ge u_2(p^*,q)\ge \min\{u_2(T,q),u_2(B,q)\}=v_2.$
Hence every Nash equilibrium pay-off is weakly above $(v_1,v_2)$. Therefore
the assumed strict Pareto dominance implies $u_1(p,q)>v_1$ and $u_2(p,q)>v_2$.

Now relabel the columns, if necessary, so that $u_1(p,L)\ge u_1(p,R).$
Since $u_1(p,q)>v_1=\min\{u_1(p,L),u_1(p,R)\}$, equality cannot hold, so in
fact $u_1(p,L)>u_1(p,R)=v_1.$ Likewise, relabel the rows, if necessary, so that $u_2(T,q)\ge u_2(B,q)$,
and again strict inequality must hold: $u_2(T,q)>u_2(B,q)=v_2$.
If a row relabeling is used, replace $p$ by $1-p$; if a column relabeling is used, replace $q$ by $1-q$. To avoid cumbersome notation, continue to denote the relabeled mixing probabilities by $p$ and $q$.

With this relabeling understood, $u_1(p,q)=q\,u_1(p,L)+(1-q)\,u_1(p,R)>u_1(p,R)$, so $q>0$. Similarly,
$u_2(p,q)=p\,u_2(T,q)+(1-p)\,u_2(B,q)>u_2(B,q)$,
so $p>0$. Thus the perturbations used below are legitimate even if one of the
original maximin strategies was pure: after relabeling, the corresponding
probability may be $1$, but it is not $0$.

We next show that $b\ge d$. Suppose instead that $b<d$. For
$0<\varepsilon\le p$, let $p_\varepsilon:=p-\varepsilon$. Then, $u_1(p_\varepsilon,R)
=
(p-\varepsilon)b+(1-p+\varepsilon)d
=
u_1(p,R)+\varepsilon(d-b)
>
u_1(p,R).$ Also,
$u_1(p_\varepsilon,L)
=
(p-\varepsilon)a+(1-p+\varepsilon)c$ tends to $u_1(p,L)$ as $\varepsilon$ tends to 0.

Because $u_1(p,L)>u_1(p,R)$, it follows that for all sufficiently small
$\varepsilon>0$, $u_1(p_\varepsilon,L)>u_1(p,R).$
Hence for such $\varepsilon$,
$\min\{u_1(p_\varepsilon,L),u_1(p_\varepsilon,R)\}>u_1(p,R)=v_1$, contradicting the fact that $p$ is a maximin strategy. Therefore $b\ge d$. By the symmetric argument, $g\ge h$.

Now consider the pure profile $(T,L)$. If $a\ge c$ and $e\ge f$, then $(T,L)$ is a Nash equilibrium. Moreover,
$u_1(p,q)=q\,u_1(p,L)+(1-q)\,u_1(p,R)\le u_1(p,L)\le a$, because $u_1(p,L)>u_1(p,R)$ and $u_1(p,L)=pa+(1-p)c\le a$ when $a\ge c$.
Similarly, $u_2(p,q)=p\,u_2(T,q)+(1-p)\,u_2(B,q)\le u_2(T,q)\le e$,
because $u_2(T,q)>u_2(B,q)$ and $u_2(T,q)=qe+(1-q)f\le e$ when $e\ge f$.
Thus $u(p,q)$ does not strictly Pareto dominate $u(T,L)$, contrary to assumption. So at least one of the inequalities $a\ge c$ and $e\ge f$ must fail.

If $c>a$, then $B$ is a best response to $L$. Together with $g\ge h$,
this implies that $(B,L)$ is a Nash equilibrium. Also,
$u_1(p,q)\le u_1(p,L)=pa+(1-p)c\le c=u_1(B,L)$,
so $u(p,q)$ does not strictly Pareto dominate $u(B,L)$.

If $f>e$, then $R$ is a best response to $T$. Together with $b\ge d$,
this implies that $(T,R)$ is a Nash equilibrium. Also, $u_2(p,q)\le u_2(T,q)=qe+(1-q)f\le f=u_2(T,R)$,
so $u(p,q)$ does not strictly Pareto dominate $u(T,R)$.

In every case we obtain a contradiction. Therefore $u(p,q)$ cannot strictly Pareto dominate all Nash equilibria.
\end{proof}

\begin{proof}[Proof of \cref{thm:dimensional-robustness}]
	Suppose first that one player has only one action. The other player's maximin strategies are exactly her best responses to that action, and the unique strategy of the first player is both maximin and part of every Nash equilibrium. Hence the maximin profile set and Nash equilibrium set coincide, so $\mathcal C_M^{m,n}$ is empty. If both players have at least two actions but the condition in part \textup{(iv)} fails, then $m=n=2$, and $\mathcal C_M^{2,2}$ is empty by \cref{lemma:2x2-maximin-no-strict-pareto}, which applies to arbitrary two-player $2\times2$ games, including asymmetric games and mixed maximin profiles. This proves that each of parts \textup{(i)}--\textup{(iii)} implies part \textup{(iv)}.
	
	The $2\times3$ game shown in the remark immediately following \cref{lemma:2x2-maximin-no-strict-pareto} has a unique maximin profile that strictly Pareto dominates the unique Nash equilibrium. The game therefore belongs to $\mathcal U_M^{2,3}$. Interchanging the players gives a game in $\mathcal U_M^{3,2}$. By \cref{lem:rectangular-padding}, $\mathcal U_M^{m,n}$ is non-empty whenever part \textup{(iv)} holds, proving \textup{(iv)}$\Rightarrow$\textup{(i)}.
	
	Part \textup{(i)} implies part \textup{(ii)} because $\mathcal U_M^{m,n}$ is open and is contained in $\mathcal C_M^{m,n}$. Part \textup{(ii)} implies part \textup{(iii)} because every non-empty open subset of $\R^{2mn}$ contains an open ball, and every open ball has positive $2mn$-dimensional Lebesgue measure. Finally, part \textup{(iii)} implies part \textup{(ii)} by \cref{thm:main-robustness}\textup{(ii)}: the portion of $\mathcal C_M^{m,n}$ outside its interior is null. All four statements are therefore equivalent.
\end{proof}

\begin{proof}[Proof of \cref{rem:nash-robustness}]
	If one player has only one action, the maximin profile set and Nash equilibrium set coincide by the argument in the preceding proof, so $\mathcal C_N^{m,n}$ is empty. For the converse, consider the $2\times2$ game
	\[
    \begin{array}{c@{~~} | @{~~}c@{~~}c}
		& L & R\\ \hline
		T & 5,5 & 3,3\\
		B & 2,1 & 2,3.
	\end{array}
	\]
	The row player's unique maximin action is $T$, and the column player's unique maximin action is $R$, so the unique maximin profile is $(T,R)$ and yields $(3,3)$. The action $T$ strictly dominates $B$, and $L$ is the column player's unique best response to $T$. Hence $(T,L)$ is the unique Nash equilibrium and yields $(5,5)$. The game belongs to $\mathcal U_N^{2,2}$. The padding construction in \cref{lem:rectangular-padding} extends it to every pair $m,n\geq2$. The openness and measure conclusions follow from \cref{thm:main-robustness}.
\end{proof}

\begin{proof}[Proof of \cref{cor:robustness-cardinality}]
	The class $\mathcal C_M^{2,3}$ contains the non-empty open set $\mathcal U_M^{2,3}$, and $\mathcal C_N^{2,2}$ contains the non-empty open set $\mathcal U_N^{2,2}$. A non-empty open subset of a finite-dimensional Euclidean space contains an open ball; such a ball contains a non-degenerate line segment and therefore has cardinality $2^{\aleph_0}$. Each global dominance class consequently has cardinality at least $2^{\aleph_0}$. Under the natural identification of finite action sets with $\{1,\ldots,m\}$ and $\{1,\ldots,n\}$, the class $\mathcal G$ is a countable union of finite-dimensional Euclidean spaces and has cardinality $2^{\aleph_0}$. Since $\mathcal C_M,\mathcal C_N\subseteq\mathcal G$, the shown equalities follow.
\end{proof}

\section{Supplementary results}

\subsection{Nash-maximin reversal theorems}

\begin{definition}[Weakly generic game]
\label{def:weakly-generic}
A game is \emph{weakly generic} if, for each player $i$ and each pure action $a_{-i}\in A_{-i}$ of the opponent, $\argmax_{a_i\in A_i}u_i(a_i,a_{-i})$
contains exactly one action.
\end{definition}

Weak genericity only rules out a tie for the best response to a pure opponent action. It allows ties among lower pay-offs and is therefore weaker than requiring all of a player's pay-offs to be distinct. Its role in the pathwise results below is to make the constructed pure maximin strategies unique.

We now move from full characterisations to sufficiency conditions. When several games are compared, a superscript $t$ is added to the notation from Section 4; thus $u_i^t$, $M_i^t$, and $v_i^t$ refer to the pay-off function, maximin set, and security level in the game $G^t$. A family $\{G^t:t\in[0,1]\}$ is called \emph{affine} if every pure-profile pay-off $u_i^t(a_1,a_2)$ is an affine function of $t$.

As mentioned in the main text, strategic equivalence alone need not preserve equilibrium pay-offs, maximin strategies, or security levels. The constructions below use the opponent-dependent terms to control these objects while leaving strategic incentives unchanged. The affine reversal theorems impose the strongest invariance requirements by retaining the same maximin strategies and security levels throughout a family; the selected-equilibrium reversal additionally fixes one equilibrium pay-off vector along the whole path.

	\begin{theorem}[Nash-maximin reversal]
	\label{thm:selected-equilibrium-reversal}
	Suppose $G$ is weakly generic and has a Nash equilibrium $\bar s^N=(\bar s_1^N,\bar s_2^N)$ satisfying $|\supp(\bar s_i^N)|\geq3.$ for all $i\in N$.
	Then there exist a pure profile $a^M=(a_1^M,a_2^M)$, numbers $\nu_i<u_i(\bar s^N)$, and an affine family of games $\{G^t:t\in[0,1]\}$ such that, for every $t\in[0,1]$:
	\begin{enumerate}[label=\textup{(\roman*)}]
	\item $G^t$ is strategically equivalent to $G$;
	\item the selected equilibrium retains its original pay-off vector,
	$u_i^t(\bar s^N)=u_i(\bar s^N)$ for all $i \in N $;
	
	\item each player has the same unique pure maximin strategy throughout the family, and the security levels are fixed, $M_i^t=\{a_i^M\}$ and $v_i^t=\nu_i$ for all $i \in N $;
	\item 
	$u^0(a^M)$ strictly Pareto dominates $u^0(\bar s^N)$ and 
	$u^1(\bar s^N)$  strictly Pareto dominates $u^1(a^M)$.
	\end{enumerate}
	\end{theorem}

    \begin{proof}[Proof of \cref{thm:selected-equilibrium-reversal}]
The strict-dominance characterisation in Theorem 2 already supplies a strategically equivalent game in which a maximin profile strictly Pareto dominates the selected equilibrium. To obtain both rankings along one affine family while preserving the selected equilibrium pay-off and fixed security levels, we use a zero-mean compensation direction.

By Lemma 2, choose pure profiles $a^M=(a_1^M,a_2^M)$ and $b=(b_1,b_2)$ such that $a_i^M$ is player $i$'s unique best response to $b_{-i}$ and $a_i^M\neq b_i$ for both players. For each player $i$, choose
\[
k_{-i}\in\supp(\bar s_{-i}^N)\setminus\{b_{-i},a_{-i}^M\},
\]
which is possible because $|\supp(\bar s_{-i}^N)|\geq3$. Write $U_i:=u_i(\bar s^N)$ and $\rho_i:=u_i(a_i^M,\bar s_{-i}^N).$

We first construct a base target pay-off function $w_i^*\colon A_{-i}\to\R$. Choose $\nu_i<U_i$ sufficiently low, set $w_i^*(b_{-i})=\nu_i$ and $w_i^*(a_{-i}^M)=U_i$,
set $w_i^*(c_{-i})=\nu_i+1$ for $c_{-i}\notin\{b_{-i},a_{-i}^M,k_{-i}\}$,
and choose the remaining coordinate $w_i^*(k_{-i})$ so that
\begin{equation}
\sum_{c_{-i}\in A_{-i}}
\bar s_{-i}^N(c_{-i})w_i^*(c_{-i})=\rho_i.
\label{eq:selected-base-mean}
\end{equation}
This choice can be made with $w_i^*(c_{-i})>\nu_i$ for every $c_{-i}\neq b_{-i}.$
Indeed,
\begin{align*}
\bar s_{-i}^N(k_{-i})\bigl[w_i^*(k_{-i})-\nu_i\bigr]
={}&\rho_i-\bar s_{-i}^N(a_{-i}^M)U_i
-\bigl(1-\bar s_{-i}^N(a_{-i}^M)\bigr)\nu_i -\sum_{c_{-i}\notin\{b_{-i},a_{-i}^M,k_{-i}\}}\bar s_{-i}^N(c_{-i}).
\end{align*}
Because $\bar s_{-i}^N(k_{-i})>0$, the coefficient $1-\bar s_{-i}^N(a_{-i}^M)$ is positive, so the right-hand side becomes positive when $\nu_i$ is sufficiently low.

Define a direction $q_i\colon A_{-i}\to\R$ by
\[
q_i(a_{-i}^M)=1,
\qquad
q_i(k_{-i})=-\frac{\bar s_{-i}^N(a_{-i}^M)}{\bar s_{-i}^N(k_{-i})},
\qquad
q_i(c_{-i})=0
\quad\text{otherwise}.
\]
Then $q_i(b_{-i})=0$ and $\sum_{c_{-i}\in A_{-i}}\bar s_{-i}^N(c_{-i})q_i(c_{-i})=0.$
Choose $\varepsilon_i>0$ sufficiently small that
\[
w_i^*(c_{-i})+(1-2t)\varepsilon_i q_i(c_{-i})>\nu_i
\]
for every $c_{-i}\neq b_{-i}$ and every $t\in[0,1]$. Define $w_i^t:=w_i^*+(1-2t)\varepsilon_i q_i$
and implement it by
\[
h_i^t(c_{-i}):=w_i^t(c_{-i})-u_i(a_i^M,c_{-i}).
\]
The resulting family is affine and strategically equivalent to $G$.

By \eqref{eq:selected-base-mean} and the zero expectation of $q_i$ under $\bar s_{-i}^N$,
\[
\sum_{c_{-i}\in A_{-i}}\bar s_{-i}^N(c_{-i})h_i^t(c_{-i})=0,
\]
so $u_i^t(\bar s^N)=U_i$ for every $t.$
Moreover, $w_i^t(b_{-i})=\nu_i$ and $w_i^t(c_{-i})>\nu_i$ for every $c_{-i}\neq b_{-i}$. Since $a_i^M$ is the unique best response to $b_{-i}$, Lemma 5 gives $M_i^t=\{a_i^M\}$ and $v_i^t=\nu_i$ for every $t$.
Finally, $u_i^t(a^M)=w_i^t(a_{-i}^M)=U_i+(1-2t)\varepsilon_i.$
Thus $a^M$ strictly Pareto dominates $\bar s^N$ at $t=0$, while $\bar s^N$ strictly Pareto dominates $a^M$ at $t=1$.
\end{proof}

\begin{example}
\label{ex:symmetric-3x3-reversal}
Starting from Figure~1, add \((-4,3,3)\) to player~1's pay-offs in columns \(x,y,z\), respectively, with the symmetric adjustments to player~2, to obtain $G^M$ below. From $G^M$, add \((0,1,-1)\) in the same way to obtain $G^N$.
	\[
	\begin{array}{c@{\qquad\qquad}c}
		\begin{array}{c@{~~} | @{~~}c@{~~}c@{~~}c}
			& x & y & z\\ \hline
			x & 3,3  & 12,0  & 3,4\\
			y & 0,12 & 11,11 & 10,7\\
			z & 4,3  & 7,10  & 6,6
		\end{array}
		&
		\begin{array}{c@{~~} | @{~~}c@{~~}c@{~~}c}
			& x & y & z\\ \hline
			x & 3,3  & 13,0  & 2,4\\
			y & 0,13 & 12,12 & 9,8\\
			z & 4,2  & 8,9   & 5,5
		\end{array}
		\\[0.4em]
		G^M & G^N
	\end{array}
	\]

	Both games are strategically equivalent to the original game and therefore have the same unique Nash equilibrium
	$
	\bar s
	=
	\left[
	\left(\frac12,\frac14,\frac14\right),
	\left(\frac12,\frac14,\frac14\right)
	\right]$ and $
	u_i(\bar s)=\frac{21}{4}.
	$
	In both games, \(z\) is the unique maximin action and the security level is \(4\). However,
	$
	u_i^M(z,z)=6>\frac{21}{4}$ and $
	u_i^N(z,z)=5<\frac{21}{4}.
	$
	Thus mutual maximin strictly Pareto dominates Nash equilibrium in $G^M$, whereas Nash equilibrium strictly Pareto dominates mutual maximin in $G^N$.

	More generally, adding \((0,t,-t)\), \(t\in[0,1]\), to $G^M$ preserves the equilibrium pay-off, the unique maximin action, and the security level, while
	$
	u_i^t(z,z)=6-t.
	$
	The ranking reverses at \(t=3/4\).
\end{example}

	\begin{theorem}[Nash-maximin general reversal]
	\label{thm:full-nash-maximin-reversal}
	Suppose $G$ is weakly generic and every Nash equilibrium $s^N=(s_1^N,s_2^N)\in\NE(G)$ satisfies
	$|\supp(s_i^N)|\geq3$ for all $i\in N$.
	For every prescribed pair $(\nu_1,\nu_2)\in\R^2$, there exist a pure profile $a^M=(a_1^M,a_2^M)$ and an affine family of games $\{G^t:t\in[0,1]\}$ such that, for every $t\in[0,1]$:
	\begin{enumerate}[label=\textup{(\roman*)}]
	\item $G^t$ is strategically equivalent to $G$;
	\item each player has the same unique pure maximin strategy throughout the family, and the prescribed security levels are fixed,
	$M_i^t=\{a_i^M\}$ and $v_i^t=\nu_i$ for all $i \in N $;
	\item for every $s^N\in\NE(G)$, $u^0(a^M)\text{ strictly Pareto dominates }u^0(s^N)$, whereas $u^1(s^N)$ strictly Pareto dominates $u^1(a^M)$.
	
	\end{enumerate}
	\end{theorem}

    \begin{proof}
	The maximin-dominant endpoint follows from Remark 5. To obtain that endpoint and the opposite uniform ranking in one affine family with fixed security levels, use the following one-parameter construction.
	
	By Lemma 2, choose pure profiles $a^M=(a_1^M,a_2^M)$ and $b=(b_1,b_2)$ such that $a_i^M$ is player $i$'s unique best response to $b_{-i}$ and $a_i^M\neq b_i$ for both players. For $s\in\NE(G)$, define $D_i(s):=u_i(s)-u_i(a_i^M,s_{-i})\geq0,$ $A_i(s):=1-s_{-i}(a_{-i}^M),$ and $B_i(s):=D_i(s)+\sum_{c_{-i}\notin\{b_{-i},a_{-i}^M\}}s_{-i}(c_{-i}).$ Every equilibrium strategy has support of at least three actions. Hence $A_i(s)>0$ and $ B_i(s)>0$ for every $s\in\NE(G)$. By Lemma 3, the Nash set is compact, so the continuous ratio $\theta_i(s):=\frac{B_i(s)}{A_i(s)}$ attains a strictly positive minimum and a finite maximum. Choose $0<\varepsilon_i<\min_{s\in\NE(G)}\theta_i(s)$, $ H_i>\max_{s\in\NE(G)}\theta_i(s),$ and define $\gamma_i(t):=(1-t)H_i+t\varepsilon_i.$
	
	For each $t\in[0,1]$, define the target pay-off function
	\[
	w_i^t(c_{-i})=
	\begin{cases}
		\nu_i, & c_{-i}=b_{-i},\\
		\nu_i+\gamma_i(t), & c_{-i}=a_{-i}^M,\\
		\nu_i+1, & c_{-i}\notin\{b_{-i},a_{-i}^M\}.
	\end{cases}
	\]
	Implement $w_i^t$ by adding $h_i^t(c_{-i}):=w_i^t(c_{-i})-u_i(a_i^M,c_{-i})$ to player $i$'s pay-offs. The family is affine and strategically equivalent to $G$. Since $w_i^t(b_{-i})=\nu_i$, all other target pay-offs are strictly above $\nu_i$, and $a_i^M$ is the unique best response to $b_{-i}$, Lemma 5 gives $M_i^t=\{a_i^M\}$ and $ v_i^t=\nu_i\text{ for every }t.$
	
	For any $s\in\NE(G)$, Lemma 4 yields
	\[
		u_i^t(a^M)-u_i^t(s)
		=\gamma_i(t)\bigl[1-s_{-i}(a_{-i}^M)\bigr]
		-D_i(s)-\sum_{c_{-i}\notin\{b_{-i},a_{-i}^M\}}s_{-i}(c_{-i}) =A_i(s)\bigl[\gamma_i(t)-\theta_i(s)\bigr].
	\]
	At $t=0$, $\gamma_i(0)=H_i>\theta_i(s)$ for every equilibrium, so $a^M$ strictly Pareto dominates the entire Nash set. At $t=1$, $\gamma_i(1)=\varepsilon_i<\theta_i(s)$ for every equilibrium, so every Nash equilibrium strictly Pareto dominates $a^M$. This proves the general reversal.
\end{proof}

The main additional assumption in \cref{thm:full-nash-maximin-reversal} is that the three-action support condition holds at every Nash equilibrium rather than only at one selected equilibrium. One family must now work uniformly over the entire Nash set.

The one-sided characterisations above already provide the maximin-dominant endpoints of the two affine reversal results. The affine theorems add the opposite ranking and require both endpoints to lie on a single path with the same pure maximin strategies and security levels throughout. \Cref{thm:selected-equilibrium-reversal} additionally fixes the original pay-off vector of the selected equilibrium along that path, whereas \cref{thm:full-nash-maximin-reversal} reverses the comparison uniformly relative to every equilibrium. Thus the two affine reversal theorems remain complementary: one gives the stronger selected-equilibrium invariance conclusion, and the other gives the stronger dominance conclusion over the entire Nash set.

\subsection{Sufficiency, characterisation, and impossibility results}
\label{subsec:sufficiency}
\subsubsection{When equilibrium pay-offs equal security levels}
	
We now study what follows when an equilibrium pay-off coincides with a security level.

	\begin{lemma}[Security attainment implies maximin is a best response]
		\label{lem:equality-security}
		Let $s^{N}$ be a Nash equilibrium and fix a player $i$. If
		$
		u_i(s^{N})=v_i,
		$
		then every maximin strategy $s^M_i\in M_i$ satisfies
		$
		u_i(s^{M}_i,s_{-i}^N)=v_i$
		and 
		$s^{M}_i\in \BR_i(s_{-i}^N).
		$
	\end{lemma}

    \begin{proof}
	\label{app:lem:equality-security}
	Let $s^M_i\in M_i$. Because $s^M_i$ is maximin, $u_i(s^M_i,s_{-i}^{N}) \ge \min_{t_{-i}}u_i(s^M_i,t_{-i}) =v_i.$ Because $s^{N}$ is a Nash equilibrium, $s_i^{N}$ is a best response to $s_{-i}^{N}$. Hence $u_i(s^M_i,s_{-i}^{N}) \le u_i(s_i^{N},s_{-i}^{N}) =u_i(s^{N}).$ By assumption, $u_i(s^{N})=v_i$. Therefore $u_i(s^M_i,s_{-i}^{N})\le v_i.$ Together with the first inequality this gives $u_i(s^M_i,s_{-i}^{N})=v_i.$ Since $s_i^{N}$ is a best response against $s_{-i}^{N}$, the same is true of $s^M_i$, so $s^M_i\in\BR_i(s_{-i}^{N})$.
\end{proof}

	If an equilibrium gives player $i$ exactly her security level, then her maximin strategy must be a best response against the opponent's equilibrium strategy.

	\begin{proposition}[Security attainment and minimax strategies]
		\label{prop:security-attainment-dual}
		Let $s^N$ be a Nash equilibrium and fix a player $i$. The following are equivalent:
		\begin{enumerate}[label=\textup{(\roman*)}]
			\item $u_i(s^N)=v_i$  
			\item $s_{-i}^N\in
			\argmin_{t_{-i}\in S_{-i}}
			\max_{t_i\in S_i}u_i(t_i,t_{-i}).$
		\end{enumerate}
	\end{proposition}

    \begin{proof}
	\label{app:prop:security-attainment-dual}
	
	Assume \textup{(i)}. Since $s^N$ is a Nash equilibrium, we have $s_i^N\in \BR_i(s_{-i}^N)$, so $\max_{t_i\in S_i}u_i(t_i,s_{-i}^N)
	=
	u_i(s_i^N,s_{-i}^N)
	=
	u_i(s^N)
	=
	v_i$.
	By Remark 1, $\underline v_i = \overline v_i = v_i$. Hence $$\max_{t_i\in S_i}u_i(t_i,s_{-i}^N)=v_i,$$ which shows that $s_{-i}^N$ attains the minimum in the definition of $v_i$. Thus \textup{(ii)} holds.
	
	Assume next \textup{(ii)}. Then $\max_{t_i\in S_i}u_i(t_i,s_{-i}^N)=v_i.$ Since $s_i^N\in \BR_i(s_{-i}^N)$, $u_i(s^N) = u_i(s_i^N,s_{-i}^N) = \max_{t_i\in S_i}u_i(t_i,s_{-i}^N) = v_i,$ so \textup{(i)} holds. Now suppose the equivalent conditions hold, and let $s^{M}_i\in M_i$. Then
	$
	u_i(s^{M}_i,s_{-i}^N)
	\ge
	\min_{t_{-i}\in S_{-i}}u_i(s^{M}_i,t_{-i})$ $
	=
	v_i,
	$
	and since $s_{-i}^N$ is minimax against player $i$, $u_i(s^{M}_i,s_{-i}^N) \le \max_{t_i\in S_i}u_i(t_i,s_{-i}^N) = v_i.$ Hence $u_i(s^{M}_i,s_{-i}^N)=v_i = \max_{t_i\in S_i}u_i(t_i,s_{-i}^N).$ Therefore $s^{M}_i\in \BR_i(s_{-i}^N)$.
\end{proof}

	Reaching the security level at equilibrium is equivalent to the opponent's equilibrium strategy being minimax against player $i$.

	\begin{definition}[Strict and quasi-strict equilibrium]
		A Nash equilibrium $s^N$ is strict if, for each $i$, $s_i^N$ is the unique best response to $s_{-i}^N$. For $s_{-i}\in S_{-i}$, let $\BR_i^{\mathrm{pure}}(s_{-i})
		:=
		\argmax_{a_i\in A_i}u_i(a_i,s_{-i}).$
		A Nash equilibrium $s^N$ is \emph{quasi-strict} if
		$
		\BR_i^{\mathrm{pure}}(s_{-i}^N)\subseteq \supp(s_i^N)
		$ for all $i \in N$.
		Equivalently, since every pure action in the support of a mixed best response is itself a pure best response,
		\[
		\BR_i^{\mathrm{pure}}(s_{-i}^N)=\supp(s_i^N)
		\qquad(i=1,2).
		\]
	\end{definition}
	
	If $s^N$ is a quasi-strict equilibrium, it rules out unused best replies. A player may mix across several equilibrium actions, but any pure action that is a best response against the opponent's equilibrium strategy must already receive positive probability in the equilibrium $s^N$.

	\begin{proposition}[Quasi-strict support contains maximin]
		\label{prop:quasi-strict-support}
		Let $s^{N}$ be a quasi-strict Nash equilibrium. Fix a player $i$ and assume $
		u_i(s^{N})=v_i.$
		Then every maximin strategy $s^M_i\in M_i$ satisfies
		$
		\supp(s^M_i)\subseteq \supp(s_i^{N}).
		$
	\end{proposition}

    \begin{proof}
	By \cref{lem:equality-security}, every $s^M_i\in M_i$ belongs to $\BR_i(s_{-i}^{N})$.
	
	If a mixed strategy is a best response to a fixed opponent strategy, then every pure action in its support is itself a pure best response to that same opponent strategy. Otherwise the mixture would place positive probability on an action whose pay-off is strictly below the best response value, which would lower the mixture's expected pay-off below that value.
	
	Therefore $\supp(s^M_i)\subseteq \BR_i^{\mathrm{pure}}(s_{-i}^{N}).$ Since $s^{N}$ is quasi-strict, $\BR_i^{\mathrm{pure}}(s_{-i}^{N})\subseteq \supp(s_i^{N}).$ Combining the two inclusions gives the claim.
\end{proof}
    
	Quasi-strictness turns \cref{lem:equality-security} into a support restriction. Once every maximin strategy becomes a best response, quasi-strictness says that all pure best responses already lie inside the equilibrium support, so maximin strategies cannot place weight outside that support.

	\begin{corollary}
		\label{cor:strict-security}
		Suppose $s^{N}=(a_1^{N},a_2^{N})$ is a strict Nash equilibrium and
		$u_i(s^{N})=v_i$ for all $i \in N$. 
		Then, for every $i \in N$, $M_i=\{a_i^{N}\}$,
		so the unique maximin profile is $s^{N}$ itself.
	\end{corollary}

    \begin{proof}
	A strict equilibrium in a finite normal-form game must be pure: if some player used a genuinely mixed strategy, then every pure action in its support would be another best response, contradicting strictness. Hence each support is a singleton, $\supp(s_i^{N})=\{a_i^{N}\}.$ A strict equilibrium is also quasi-strict. Applying \cref{prop:quasi-strict-support} to each player gives $\supp(s^M_i)\subseteq \{a_i^{N}\}\text{ for every }s^M_i\in M_i.$ The only mixed strategy whose support is contained in $\{a_i^{N}\}$ is the degenerate strategy that assigns probability $1$ to $a_i^{N}$. Hence $a_i^{N}$ is the unique maximin for each player, and the only maximin profile is $s^{N}$ itself.
\end{proof}

	Strictness removes all alternative best responses. Combined with equality to security level, it forces each maximin set to collapse to the unique equilibrium action, so the entire maximin profile is uniquely determined.

	\begin{theorem}[Quasi-strict equilibrium at maximin profiles]
		\label{thm:quasi-strict-pins-maximin-pairs}
		Let $\bar{s}$ be a quasi-strict equilibrium and suppose $u_i(\bar{s})=v_i$ for every $i \in N$.
		Then every maximin profile $(s^M_1,s^M_2)\in M_1\times M_2$ satisfies $u_i(s^M_1,s^M_2)=v_i$ for all $i \in N$.
	\end{theorem}

\begin{proof}
	Fix $(s^M_1,s^M_2)\in M_1\times M_2$. By \cref{prop:quasi-strict-support}, $\supp(s^M_i)\subseteq \supp(\bar{s}_i)$ for all $i\in N$.	We first prove the claim for player $1$. Since $s^M_1\in M_1$, $u_1(s^M_1,a_2)\ge v_1\text{ for every }a_2\in A_2.$ By \cref{lem:equality-security}, $u_1(s^M_1,\bar{s}_2)=v_1.$ Now expand $u_1(s^M_1,\bar{s}_2)$ as a convex combination:
	$v_1
	=
	u_1(s^M_1,\bar{s}_2)
	=
	\sum_{a_2\in \supp(\bar{s}_2)}
	\bar{s}_2(a_2)\,u_1(s^M_1,a_2).$ Each term in this average is at least $v_1$, and each coefficient $\bar{s}_2(a_2)$ is strictly positive on $\supp(\bar{s}_2)$. Therefore the average can equal $v_1$ only if $u_1(s^M_1,a_2)=v_1\text{ for every }a_2\in \supp(\bar{s}_2).$ Since $\supp(s^M_2)\subseteq \supp(\bar{s}_2)$, the same equality holds for every $a_2\in \supp(s^M_2)$. Averaging with respect to $s^M_2$ now yields
	\[
	u_1(s^M_1,s^M_2)
	=
	\sum_{a_2\in \supp(s^M_2)}
	s^M_2(a_2)\,u_1(s^M_1,a_2)
	=
	\sum_{a_2\in \supp(s^M_2)}
	s^M_2(a_2)\,v_1
	=
	v_1.
	\]

	The argument for player~2 is symmetric. Since $s^M_2\in M_2$, $u_2(a_1,s^M_2)\ge v_2\text{ for every }a_1\in A_1.$ Also, $u_2(\bar{s}_1,s^M_2)=v_2,$ because $s^M_2$ guarantees at least $v_2$, while $\bar{s}_2$ is a best response to $\bar{s}_1$ and $u_2(\bar{s})=v_2.$ 
    
    Expanding again as a convex combination,
	$$v_2= u_2(\bar{s}_1,s^M_2)= 
	\sum_{a_1\in \supp(\bar{s}_1)}
	\bar{s}_1(a_1)\,u_2(a_1,s^M_2).$$
	Every term is at least $v_2$, so equality of the average forces $u_2(a_1,s^M_2)=v_2\text{ for every }a_1\in \supp(\bar{s}_1).$ Using $\supp(s^M_1)\subseteq \supp(\bar{s}_1)$, we conclude
	$$
	u_2(s^M_1,s^M_2)
	=
	\sum_{a_1\in \supp(s^M_1)}
	s^M_1(a_1)\,u_2(a_1,s^M_2)
	=
	v_2.
	$$
	Hence $u_i(s^M_1,s^M_2)=v_i$ for all $i\in N$, as desired.
\end{proof}

	Here, the support restriction translates into a pay-off restriction. Every maximin strategy must live inside the equilibrium support, and on that support the relevant expected pay-offs are already forced to equal the security levels. Therefore any maximin profile reproduces the security vector.

	\begin{corollary}
		\label{cor:all-NE-equal-security}
		Assume that every Nash equilibrium $s^N$ satisfies $u_i(s^N)=v_i$ for every $i\in N$.
		Then every maximin profile $(s^M_1,s^M_2)\in M_1\times M_2$ satisfies $	u_i(s^M_1,s^M_2)=v_i$ for all $i\in N$.
	\end{corollary}

    \begin{proof}
	By \citet{norde1999}, every two-player game has a quasi-strict equilibrium. Let $\bar{s}$ be such a quasi-strict equilibrium. By assumption, $u_i(\bar{s})=v_i$ for all $i\in N$. The conclusion now follows from \cref{thm:quasi-strict-pins-maximin-pairs}.
\end{proof}
    
	If every equilibrium already lies on the security vector, then choosing one quasi-strict equilibrium is enough to transmit that equality to all maximin profiles.

\subsubsection{When maximin profiles are Nash equilibria}

    We now start from a maximin profile and ask what follows when it is either a minimax profile or a Nash equilibrium. 

    \begin{proposition}[Maximin and minimax profile implies Nash equilibrium]
	\label{prop:maximin-plus-minimax-implies-NE}
    For each $i\in N$, suppose that $s_i^M$ is both a maximin strategy and a minimax strategy of player $i$. Then, the profile $s^M$ is a Nash equilibrium and $u(s^M)=(v_1,v_2)$.
\end{proposition}

\begin{proof}
	Because $s_2^M$ is a minimax strategy of player $2$, it attains the minimum in the definition of player $1$'s minimax value. Hence $\max_{t_1\in S_1}u_1(t_1,s_2^M)=v_1$ by Remark 1. Therefore $u_1(t_1,s_2^M)\le v_1\text{ for every }t_1\in S_1.$ On the other hand, since $s_1^M\in M_1$ is a maximin strategy, $\min_{t_2\in S_2}u_1(s_1^M,t_2)=v_1,$ so $u_1(s_1^M,t_2)\ge v_1\text{ for every }t_2\in S_2.$ Evaluating the last inequality at $t_2=s_2^M$ gives $u_1(s_1^M,s_2^M)\ge v_1,$ while evaluating the first inequality at $t_1=s_1^M$ gives $u_1(s_1^M,s_2^M)\le v_1.$ Hence $u_1(s_1^M,s_2^M)=v_1.$ Since every $t_1\in S_1$ satisfies $u_1(t_1,s_2^M)\le v_1=u_1(s_1^M,s_2^M)$, $s_1^M\in \BR_1(s_2^M)$.
	
	The argument for player $2$ is symmetric. Because $s_1^M$ is a minimax strategy of player $1$, $\max_{t_2\in S_2}u_2(s_1^M,t_2)=v_2$,
	so $u_2(s_1^M,t_2)\le v_2\text{ for every }t_2\in S_2.$ Since $s_2^M\in M_2$ is maximin, $\min_{t_1\in S_1}u_2(t_1,s_2^M)=v_2,$ hence $u_2(t_1,s_2^M)\ge v_2\text{ for every }t_1\in S_1.$ Evaluating at $t_1=s_1^M$ and $t_2=s_2^M$ yields $u_2(s_1^M,s_2^M)=v_2.$ Therefore, $s_2^M\in \BR_2(s_1^M).$ Thus, $s^M=(s_1^M,s_2^M)$ is a Nash equilibrium.
\end{proof}

	\begin{proposition}[Maximin Nash equilibria and saddle inequalities]
		\label{prop:maximin-profiles-saddle}
		Let $s^M=(s_1^M,s_2^M)$ be both a maximin profile and a Nash equilibrium. The following are equivalent:
		\begin{enumerate}[label=\textup{(\roman*)}]
			\item $u(s^M)=(v_1,v_2)$;
			\item for every $t_1\in S_1$ and $t_2\in S_2$,
			\[
			u_1(t_1,s_2^M)\le u_1(s^M)\le u_1(s_1^M,t_2),
			\]
			\[
			u_2(s_1^M,t_2)\le u_2(s^M)\le u_2(t_1,s_2^M).
			\]
		\end{enumerate}
	\end{proposition}

    \begin{proof}
	Assume \textup{(i)}. Since $s^M$ is a Nash equilibrium, $u_1(t_1,s_2^M)\le u_1(s^M)$  and  $u_2(s_1^M,t_2)\le u_2(s^M)$ for all $t_1\in S_1$ and $t_2\in S_2$. Since $s_1^M\in M_1$ and $s_2^M\in M_2$ are maximin strategies, $u_1(s_1^M,t_2)\ge v_1=u_1(s^M)$ and $ u_2(t_1,s_2^M)\ge v_2=u_2(s^M)$ for all $t_1\in S_1$ and $t_2\in S_2$. Hence \textup{(ii)} holds.
	
	Assume \textup{(ii)}. Then $u_1(s^M)\le \min_{t_2\in S_2}u_1(s_1^M,t_2)=v_1,$ because $s_1^M\in M_1$ is a maximin, and $u_2(s^M)\le \min_{t_1\in S_1}u_2(t_1,s_2^M)=v_2,$ because $s_2^M\in M_2$ is a maximin. On the other hand, since $s^M$ is a Nash equilibrium, Remark 2 gives $u_i(s^M)\ge v_i$ for all $i\in N$. Therefore $u(s^M)=(v_1,v_2),$ so \textup{(i)} holds.
\end{proof}

If a maximin profile is also a Nash equilibrium, then it realises the security vector exactly when each player's opponent is minimax against her. Equivalently, the profile is then a simultaneous ``saddle point'' of both pay-off functions.
    
	\begin{proposition}[When every maximin profile is a security-realizing equilibrium]
		\label{prop:maximin-rectangle}
		The following are equivalent:
		\begin{enumerate}[label=\textup{(\roman*)}]
			\item every maximin profile $(s_1^M,s_2^M)\in M_1\times M_2$ is a Nash equilibrium and satisfies $u(s_1^M,s_2^M)=(v_1,v_2)$;
			\item
			\[
			M_1\subseteq \argmin_{t_1\in S_1}\max_{t_2\in S_2}u_2(t_1,t_2)
			\qquad\text{and}\qquad
			M_2\subseteq \argmin_{t_2\in S_2}\max_{t_1\in S_1}u_1(t_1,t_2).
			\]
		\end{enumerate}
	\end{proposition}

\begin{proof}
	Assume \textup{(ii)} and let $(s_1^M,s_2^M)\in M_1\times M_2$ be a maximin profile. Since $s_2^M$ is minimax against player $1$, $\max_{t_1\in S_1}u_1(t_1,s_2^M)=v_1$ by Remark 1. Hence $u_1(t_1,s_2^M)\le v_1\text{ for every }t_1\in S_1.$ Since $s_1^M\in M_1$, $\min_{t_2\in S_2}u_1(s_1^M,t_2)=v_1,$ so $$u_1(s_1^M,t_2)\ge v_1\text{ for every }t_2\in S_2.$$ Evaluating at $t_1=s_1^M$ and $t_2=s_2^M$ yields $u_1(s_1^M,s_2^M)=v_1.$ Therefore $s_1^M\in \BR_1(s_2^M)$. By symmetry, $u_2(s_1^M,s_2^M)=v_2$ and $ s_2^M\in \BR_2(s_1^M).$ Hence $(s_1^M,s_2^M)$ is a Nash equilibrium, and \textup{(i)} holds.
	
	Assume \textup{(i)}. Let $s_2^M\in M_2$ be a maximin strategy. Because $M_1\neq\varnothing$, choose any $s_1^M\in M_1$. Then $(s_1^M,s_2^M)$ is a Nash equilibrium and $u_1(s_1^M,s_2^M)=v_1.$ By \cref{prop:security-attainment-dual},
	$
	s_2^M\in \argmin_{t_2\in S_2}\max_{t_1\in S_1}u_1(t_1,t_2).
	$
	Since $s_2^M$ was arbitrary, $$M_2\subseteq \argmin_{t_2\in S_2}\max_{t_1\in S_1}u_1(t_1,t_2).$$ The proof that
	$
	M_1\subseteq \argmin_{t_1\in S_1}\max_{t_2\in S_2}u_2(t_1,t_2)
	$
	is symmetric.
\end{proof}

	This result provides a set-valued analogue of \cref{prop:maximin-profiles-saddle}. Rather than asking when a particular maximin profile is an exact security equilibrium, it characterises when every maximin profile does so. The condition is that each player's entire maximin set lies within the opponent-side minimax set against the other player.

\subsection{\texorpdfstring{$n$}{n}-player extensions}
\label{app:nplayer}

For player \(i\), let \(A_{-i}:=\prod_{j\neq i}A_j\), let \(S_{-i}^{I}:=\prod_{j\neq i}S_j\) denote the independent strategies of the other players, and let \(S_{-i}^{C}:=\Delta(A_{-i})\) denote arbitrary correlated distributions over their joint actions. For every \(s_i\in S_i\),
$\min_{s_{-i}\in S_{-i}^{I}}u_i(s_i,s_{-i})
=
\min_{\mu_{-i}\in S_{-i}^{C}}u_i(s_i,\mu_{-i})
=
\min_{a_{-i}\in A_{-i}}u_i(s_i,a_{-i}).$
Both strategy sets contain every degenerate distribution on a pure joint profile, while every mixed or correlated pay-off is a convex combination of pure-profile pay-offs. Hence the two definitions yield the same maximin values and strategies. A maximin profile remains a product profile in \(\prod_i M_i\).

Strategic equivalence becomes $
\widetilde u_i(a_i,a_{-i})
=
u_i(a_i,a_{-i})+h_i(a_{-i}),$
where \(h_i\) is defined on the full joint action set of the other players. The added term still cancels from every unilateral pay-off comparison. Therefore, Lemma 1 and Theorem 1 extend unchanged. In particular, the proof of Theorem 1 again chooses a pure best response to the selected equilibrium strategy of the other players and assigns that action the constant target pay-off \(u_i(\bar s)\).

The characterisations in Theorems 2 and 3 also extend, with \(a_{-i}\) interpreted as a joint action profile, \(r_i\colon A_{-i}\to\R_+\), and the supporting belief taken from \(\Delta(A_{-i})\). Correlation is required for necessity because the dual problem treats the other players as a single adversarial coalition; restricting the supporting belief to independent strategies remains sufficient but is generally not necessary.

The reversal results in \cref{thm:selected-equilibrium-reversal,thm:full-nash-maximin-reversal} remain valid under the same weak genericity requirement after replacing the componentwise support assumptions by $\lvert\supp(\bar s_{-i})\rvert\geq3 $ for every $i$,
with the corresponding condition imposed at every Nash equilibrium in \cref{thm:full-nash-maximin-reversal}.

Finally, Theorem 4 extends unchanged to the corresponding finite-dimensional \(n\)-player game space. Mutatis mutandis, the arguments in the proof apply with player \(i\) facing the joint action set \(A_{-i}\). The \(n\)-player analogue of Theorem 5 also holds after replacing the two-player action threshold by its dimensional counterpart: strict maximin dominance occurs on a non-empty open set whenever at least three players have at least two actions, or exactly two players have at least two actions and one of them has at least three. Thus, in \(n\)-player games, a third player can play the same dimensional role as a third action in the two-player case.

\subsection{Economic applications}
\label{sec:economic-applications}

In this section, we compare mutual maximin behaviour with Nash equilibrium in several familiar economic environments: Cournot competition with capacity constraints, a guns-versus-butter conflict model \citep{garfinkel2007}, and an interdependent-security model related to the system-reliability framework of \citep{varian2004}.

\subsubsection{Cournot with capacity constraints}
\label{subsec:capacity-cournot-application}

Two firms choose quantities $q_i\in[0,K]$. Inverse demand is $P(Q)=a-bQ$, where $b>0$, marginal cost is $c\geq0$, and $a-c>0$. Firm $i$'s profit is $ \pi_i(q_i,q_{-i})=q_i\bigl(a-c-bq_i-bq_{-i}\bigr).$
The linear specification is used on the feasible action set. To ensure non-negative prices, assume $2K\leq a/b$.

\begin{proposition}
\label{prop:application-cournot}
The unique maximin action and its associated security level are given by
\begin{equation}
\left(q^M, v_i\right)=
\begin{cases}
\left(K, K(a-c-2bK)\right), & 0<K\leq \dfrac{a-c}{3b},\\[0.7em]
\left(\dfrac{a-c-bK}{2b}, \dfrac{(a-c-bK)^2}{4b}\right), & \dfrac{a-c}{3b}<K<\dfrac{a-c}{b},\\[0.9em]
(0,0), & K\geq \dfrac{a-c}{b}.
\end{cases}
\label{eq:cournot-maximin-action}
\end{equation}
The unique Nash equilibrium is symmetric, with $
q^N=\min\left\{K,\frac{a-c}{3b}\right\}$. Consequently, mutual maximin and Nash coincide for $K\leq(a-c)/(3b)$; mutual maximin strictly Pareto dominates Nash for $\frac{a-c}{3b}<K<\frac{2(a-c)}{3b}$; the two profiles yield equal profits at $K=2(a-c)/(3b)$; and Nash strictly Pareto dominates mutual maximin for $K>2(a-c)/(3b)$.
\end{proposition}

\begin{proof}
Profit is decreasing in the opponent's quantity, so the worst feasible opponent action is $q_{-i}=K$. The security problem is therefore $\max_{q_i\in[0,K]}q_i(a-c-bK-bq_i)$.
The objective is strictly concave and has unconstrained maximiser $(a-c-bK)/(2b)$. Imposing the constraint $q_i\in[0,K]$, the maximiser is $K$ if this expression is at least $K$, $0$ if it is non-positive, and the expression itself otherwise. This gives \eqref{eq:cournot-maximin-action}.

The best response is $\BR_i(q_{-i})=\min\left\{K,\max\left\{0,\frac{a-c-bq_{-i}}{2b}\right\}\right\}$.
It is $1/2$-Lipschitz, so the joint best-response map is a contraction and has a unique fixed point. Symmetry then gives the result. At the maximin profile,
\[
\pi_i(q^M,q^M)=
\begin{cases}
K(a-c-2bK), & K\leq(a-c)/(3b),\\[0.4em]
K(a-c-bK)/2, & (a-c)/(3b)<K<(a-c)/b,\\[0.4em]
0, & K\geq(a-c)/b.
\end{cases}
\]
At Nash, profit equals $K(a-c-2bK)$ when the cap binds and $(a-c)^2/(9b)$ otherwise. In the middle region, the sign is determined by the two factors $3bK-(a-c)$ and $2(a-c)-3bK$, which establish the stated ranking.
\end{proof}

The capacity cap makes the opponent's most aggressive action finite. Maximin therefore is the best response to an opponent producing $K$. When the cap is tight, it binds under both maximin and Nash, so the two rules select the same quantity. For intermediate values of $K$, mutual maximin reduces output relative to Nash, raises the market price, and increases both firms' profits. When the cap is sufficiently loose, however, maximin becomes too cautious: output falls sharply and eventually reaches zero, so Nash again yields higher profits.

\subsubsection{Guns versus butter}
\label{subsec:guns-butter-application}

Each player has an endowment $R>0$ and allocates $g_i\in[0,R]$ to arms. The remainder produces butter, so total production is $2R-g_i-g_{-i}$. If at least one player arms, player $i$ obtains the Tullock share of production, $
u_i(g_i,g_{-i})=\frac{g_i}{g_i+g_{-i}}\bigl(2R-g_i-g_{-i}\bigr)$ where $g_i+g_{-i}>0$. At $(0,0)$, production is divided equally and $u_i(0,0)=R$.

\begin{proposition}[Guns versus butter]
\label{prop:application-guns}
The unique maximin action is $g^M=R(\sqrt2-1)$, with security level $v_i=R(3-2\sqrt2)$.
The game has a unique Nash equilibrium $g_1^N=g_2^N=\frac R2$. At the two symmetric profiles $u_i(g^M,g^M)=R(2-\sqrt2)$ and $ u_i(g^N,g^N)=\frac R2$,
so
\begin{equation}
u_i(g^M,g^M)-u_i(g^N,g^N)=R\left(\frac32-\sqrt2\right)=\frac R2(\sqrt2-1)^2>0.
\label{eq:guns-gap}
\end{equation}
Thus the maximin profile strictly Pareto dominates the unique Nash equilibrium for every $R>0$.
\end{proposition}

\begin{proof}
For $g_i>0$, $
\frac{\partial u_i(g_i,g_{-i})}{\partial g_{-i}}=-\frac{2Rg_i}{(g_i+g_{-i})^2}<0$.
The worst feasible opponent action is therefore $g_{-i}=R$. Hence, the minimum pay-off that action $g$ yields against the opponent's possible actions is $
\underline u(g)
:=
\min_{g_{-i}\in[0,R]}u_i(g,g_{-i})
=
\frac{g(R-g)}{g+R}$,
where $\underline u(0)=0$.
It satisfies $\underline u'(g)=\frac{R^2-2Rg-g^2}{(g+R)^2}$ and  $\underline u''(g)=-\frac{4R^2}{(g+R)^3}<0.$
The unique maximiser is the positive root of $R^2-2Rg-g^2=0$, which gives the maximin. Substitution yields the security level.

For $g_{-i}>0$, own pay-off is strictly concave and the unique best response is $\BR_i(g_{-i})=\sqrt{2Rg_{-i}}-g_{-i}$. There is no best response to $g_{-i}=0$: every $g_i>0$ gives $2R-g_i$, whose supremum is approached as $g_i$ tends to zero, while $g_i=0$ gives only $R$. Thus an equilibrium must have both actions positive. The two first-order conditions imply $g_1=g_2$, and substitution into the best-response formula gives the Nash equilibrium.

At a symmetric positive profile each player receives one-half of total production, so $u_i(g,g)=R-g$. Evaluating at $g^M$ and $g^N$ gives \eqref{eq:guns-gap}.
\end{proof}

Arming improves a player's share of total output, but it also reduces the amount left for production. Nash players take account of the first effect but not the loss they impose on the opponent through lower total production. This creates an arms race. Maximin responds to a fully armed opponent and therefore chooses less arming than Nash. When both players choose the maximin action, each still receives one-half of total output, but more resources remain for production. Mutual maximin therefore gives both players a higher pay-off for every $R>0$.

\subsubsection{Interdependent security}
\label{subsec:security-application}

Each player chooses $x_i\in[0,1]$, interpreted as the success probability of an independent protective measure. The common system is protected whenever at least one measure succeeds, so $P(x_i,x_{-i})=1-(1-x_i)(1-x_{-i})=x_i+x_{-i}-x_i x_{-i}$.
Each player values successful protection at $V>0$ and bears quadratic cost,
\[
u_i(x_i,x_{-i})=V(x_i+x_{-i}-x_i x_{-i})-\frac c2x_i^2,\qquad c>0.
\]

\begin{proposition}[Interdependent security]
\label{prop:application-security}
The unique maximin action is
\begin{equation}
x^M=\min\left\{\frac Vc,1\right\}.
\label{eq:security-maximin}
\end{equation}
Its security level is
\begin{equation}
v_i=
\begin{cases}
\dfrac{V^2}{2c}, & 0<V\leq c,\\[0.8em]
V-\dfrac c2, & V\geq c.
\end{cases}
\label{eq:security-level}
\end{equation}
The complete Nash equilibrium set is
\begin{equation}
\NE(G)=
\begin{cases}
\left\{\left(\dfrac{V}{V+c},\dfrac{V}{V+c}\right)\right\}, & 0<V<c,\\[0.9em]
\left\{(x,1-x):x\in[0,1]\right\}, & V=c,\\[0.9em]
\left\{\left(\dfrac{V}{V+c},\dfrac{V}{V+c}\right),(1,0),(0,1)\right\}, & V>c.
\end{cases}
\label{eq:security-equilibria}
\end{equation}
The pay-off comparison is as follows:
\begin{enumerate}[label=\textup{(\roman*)}]
\item if $0<V/c<(\sqrt{17}-1)/4$, the maximin profile strictly Pareto dominates every Nash equilibrium;
\item if $V/c=(\sqrt{17}-1)/4$, the maximin profile and the unique Nash equilibrium give both players the same pay-off;
\item if $V/c>(\sqrt{17}-1)/4$, every Nash equilibrium Pareto dominates the maximin profile. The dominance is strict for both players at the symmetric equilibrium and, when $V=c$, at every non-corner equilibrium. The corner equilibria $(1,0)$ and $(0,1)$ occur for $V\geq c$; at either one, one player is indifferent and the other is strictly better off.
\end{enumerate}
\end{proposition}

\begin{proof}
Since $\frac{\partial u_i}{\partial x_{-i}}=V(1-x_i)\geq0$,
the worst opponent action is $x_{-i}=0$. The security problem is therefore $\max_{x\in[0,1]}\left\{Vx-\frac c2x^2\right\}.$
Strict concavity gives \eqref{eq:security-maximin} and \eqref{eq:security-level}.

The unique best response to an opponent action is $\BR_i(x_{-i})=\min\left\{1,\frac Vc(1-x_{-i})\right\}.$
For $V<c$, the upper bound never binds and the two best-response equations have the unique symmetric solution $V/(V+c)$. At $V=c$, they reduce to $x_1+x_2=1$. For $V>c$, the only interior solution is $(V/(V+c),V/(V+c))$, and the two additional boundary solutions are $(1,0)$ and $(0,1)$. These are the profiles in \eqref{eq:security-equilibria}.

For $0<V<c$, the maximin profile is $(V/c,V/c)$ and gives $u_i(x^M,x^M)=\frac{3V^2}{2c}-\frac{V^3}{c^2}.$
The unique equilibrium gives
\[
u_i\!\left(\frac{V}{V+c},\frac{V}{V+c}\right)=\frac{V^2\left(V+\tfrac32c\right)}{(V+c)^2}.
\]
Their difference is
\[
u_i(x^M,x^M)-u_i\!\left(\frac{V}{V+c},\frac{V}{V+c}\right)=\frac{V^3(2c^2-Vc-2V^2)}{2c^2(V+c)^2}.
\]
Its sign is the sign of $2-V/c-2(V/c)^2$, whose unique positive root is $(\sqrt{17}-1)/4$. This proves the strict-dominance, equality, and reversal statements for $V<c$.

At $V=c$, mutual maximin is $(1,1)$ and gives each player $c/2$. Every equilibrium is $(x,1-x)$, at which
\[
u_1(x,1-x)=\frac c2\left[1+(1-x)^2\right],\qquad u_2(x,1-x)=\frac c2(1+x^2).
\]
Both pay-offs are strictly above $c/2$ when $x\in(0,1)$. At either corner $x\in\{0,1\}$, one pay-off equals $c/2$ and the other equals $c$. For $V>c$, mutual maximin remains $(1,1)$ and gives $V-c/2$ to each player. At the symmetric equilibrium,
\[
u_i(1,1)-u_i\!\left(\frac{V}{V+c},\frac{V}{V+c}\right)=-\frac{c^3}{2(V+c)^2}<0.
\]
At $(1,0)$, player~1 receives $V-c/2$ and player~2 receives $V$; the profile $(0,1)$ is symmetric. This establishes the stated comparison with every equilibrium.
\end{proof}

A player's protection benefits the opponent, so the worst case for each player is that the opponent provides no protection. Maximin therefore chooses the level that would be optimal if the player had to protect the system alone. Nash players instead expect some protection from the opponent and free-ride, yielding $x^N=V/(V+c)<V/c=x^M$ when $V/c<1$.

At low values of $V/c$, Nash investment is too low because of free riding, and the extra protection under mutual maximin is worth more than its additional cost. Mutual maximin therefore strictly Pareto dominates Nash when
$V/c<(\sqrt{17}-1)/4$.
Above this threshold, the maximin action is too high when both players choose it. Nash performs better because it either moderates both players' investment or allows one player to provide all the protection.

\subsection{\texorpdfstring{$2\times 2$ impossibility}{2 x 2 impossibility}}
\label{subsec:2by2_impossibility}

Recall that by Lemma 7, no maximin profile can strictly Pareto dominate all Nash equilibria in a $2\times 2$ game. To illustrate, consider the classic games in Table~\ref{tab:2x2-benchmarks}. In the Chicken game, $(B,R)$ is the maximin profile and yields pay-offs $(2,2)$. This strictly Pareto dominates the mixed Nash equilibrium. Nevertheless, it does not strictly Pareto dominate all Nash equilibria, because $(B,L)$ and $(T,R)$ are Nash equilibria with pay-offs $(4,1)$ and $(1,4)$.

\begin{landscape}
\begin{table}[t]
\centering

\setlength{\tabcolsep}{4pt}
\renewcommand{\arraystretch}{1.15}

\begin{tabular}{@{}p{3.2cm} p{2cm} p{2.8cm} p{2.2cm} p{2cm} p{2cm}@{}}
\toprule
Game & Matrix & Nash equilibrium & NE pay-offs & Maximin pair: pay-offs & Minimax pair: pay-offs \\
\midrule

Prisoner's Dilemma &
$\begin{smallmatrix}
3,3 & 0,4\\
4,0 & 1,1
\end{smallmatrix}$ &
$(0,0)$ &
$1,1$ &
$(0,0):\,1,1$ &
$(0,0):\,1,1$ \\[0.4em]

Chicken (Hawk-Dove) &
$\begin{smallmatrix}
0,0 & 4,1\\
1,4 & 2,2
\end{smallmatrix}$ &
$(1,0);\,(0,1);\;(\tfrac{2}{3},\tfrac{2}{3})$ &
$4,1;\;1,4;\;(\tfrac{4}{3},\tfrac{4}{3})$ &
$(0,0):\,2,2$ &
$(1,1):\,0,0$ \\[0.4em]

Assurance (Stag Hunt) &
$\begin{smallmatrix}
4,4 & 1,3\\
3,1 & 2,2
\end{smallmatrix}$ &
$(1,1);\,(0,0);\;(\tfrac{1}{2},\tfrac{1}{2})$ &
$4,4;\;2,2;\;(\tfrac{5}{2},\tfrac{5}{2})$ &
$(0,0):\,2,2$ &
$(0,0):\,2,2$ \\[0.4em]

Battle of the Sexes &
$\begin{smallmatrix}
2,1 & 0,0\\
0,0 & 1,2
\end{smallmatrix}$ &
$(1,1);\,(0,0);\;(\tfrac{2}{3},\tfrac{1}{3})$ &
$2,1;\;1,2;\;(\tfrac{2}{3},\tfrac{2}{3})$ &
$(\tfrac{1}{3},\tfrac{2}{3}):\,\tfrac{2}{3},\tfrac{2}{3}$ &
$(\tfrac{2}{3},\tfrac{1}{3}):\,\tfrac{2}{3},\tfrac{2}{3}$ \\[0.4em]

Harmony &
$\begin{smallmatrix}
4,4 & 3,2\\
2,3 & 1,1
\end{smallmatrix}$ &
$(1,1)$ &
$4,4$ &
$(1,1):\,4,4$ &
$(0,0):\,1,1$ \\[0.4em]

Matching Pennies &
$\begin{smallmatrix}
1,\text{-}1 & \text{-}1,1\\
\text{-}1,1 & 1,\text{-}1
\end{smallmatrix}$ &
$(\tfrac{1}{2},\tfrac{1}{2})$ &
$0,0$ &
$(\tfrac{1}{2},\tfrac{1}{2}):\,0,0$ &
$(\tfrac{1}{2},\tfrac{1}{2}):\,0,0$ \\[0.4em]

Coordination &
$\begin{smallmatrix}
4,4 & 0,0\\
0,0 & 2,2
\end{smallmatrix}$ &
$(1,1);\,(0,0);\;(\tfrac{1}{3},\tfrac{1}{3})$ &
$4,4;\;2,2;\;(\tfrac{4}{3},\tfrac{4}{3})$ &
$(\tfrac{1}{3},\tfrac{1}{3}):\,\tfrac{4}{3},\tfrac{4}{3}$ &
$(\tfrac{1}{3},\tfrac{1}{3}):\,\tfrac{4}{3},\tfrac{4}{3}$ \\

\bottomrule
\end{tabular}

\vspace{0.3em} \textit{Notes}: A strategy profile $(p,q)$ denotes:
$1,1$ = Top/Left; $1,0$ = Top/Right; $0,1$ = Bottom/Left; $0,0$ = Bottom/Right.
\caption{Classic $2\times 2$ games}
\label{tab:2x2-benchmarks}
\end{table}
\end{landscape}

\section{\texorpdfstring{Strict Symmetric $3\times3$ Games}{Strict Symmetric 3 x 3 Games}}
\label{sec:strict-symmetric}

While it is not possible for a maximin profile to strictly Pareto dominate Nash equilibria in $2\times2$ games, we have illustrated that this is possible in $3\times3$ games. In this section, we provide a systematic analysis of the class of strictly ordinal symmetric $3\times3$ games. The aim is to obtain a sense of frequency, rather than cardinality, within a large and well-structured environment. Because pay-offs are purely ordinal, this class abstracts from cardinal utility and hence focuses only on pure strategies.

We begin with the labeled class of strict symmetric $3\times3$ games in which each player's pay-offs are a permutation of $\{1,\dots,9\}$, and player~2's pay-off matrix is the transpose of player~1's. Hence there are $9!=362{,}880$ such labeled games.

To avoid multiple counting due to relabeling of actions, we exploit the natural symmetry of the common action set. We regard two games as equivalent whenever one can be obtained from the other by a simultaneous relabeling of rows and columns that preserves symmetry. Because all entries of player 1's pay-off matrix $A$ are distinct, each equivalence class has size $3!=6$, so the $362{,}880$ labeled games reduce to $
362{,}880/6 = 60{,}480$
distinct games.

Throughout, we consider both pay-off profiles, which are vectors, and strategy guarantees associated with Nash, maximin, and minimax strategies, which are scalars. A Nash guarantee is the minimum pay-off that a Nash strategy secures; similarly, a minimax guarantee is the minimum pay-off secured by a minimax strategy.

Table~\ref{tab:ordinal-ne-counts} reports the distribution of pure Nash equilibria. Pure equilibria are common but not universal: $92.59\%$ of games admit at least one pure equilibrium, while $7.41\%$ do not. In this setting, every game admits a unique pure maximin strategy and a unique pure minimax strategy.

\begin{table}[htbp]
\centering
\small
\begin{tabular}{lcc}
\toprule
Number of pure Nash equilibria & Count & Share \\
\midrule
0 & 4,480 & 7.41\% \\
1 & 20,160 & 33.33\% \\
2 & 26,880 & 44.44\% \\
3 & 8,960 & 14.81\% \\
\bottomrule
\end{tabular}
\vspace{0.2cm}

\caption{Distribution of pure Nash equilibria in strictly ordinal symmetric $3\times3$ games.}
\label{tab:ordinal-ne-counts}
\end{table}

We now compare maximin, minimax, and Nash equilibrium in terms of realised pay-offs and guarantees, conditioning on the existence of at least one pure equilibrium. The main results are reported in Table~\ref{tab:stats}.

Two broad patterns emerge. First, as expected, maximin performs better than minimax relative to Nash equilibria. Among the $56{,}000$ games with at least one pure equilibrium, the maximin profile weakly Pareto dominates all pure equilibria in $38.58\%$ of games and strictly dominates them in $1.83\%$. By contrast, minimax weakly dominates equilibrium pay-offs in only $12.00\%$ of cases and obviously never strictly dominates them.

The comparison is particularly transparent in the subset of one-third of games with a unique pure Nash equilibrium, where there are no coordination issues. In this subset, the maximin profile coincides with the Nash equilibrium in $72.72\%$ of games and strictly dominates it in $3.86\%$. Conversely, the unique equilibrium strictly Pareto dominates the maximin profile in $23.42\%$ of cases. However, the maximin strategy guarantees a strictly higher pay-off than the Nash equilibrium in $27.28\%$ of the games. Overall, neither rule dominates the other across the board.

\begin{table}[htbp]
\centering
\small
\begin{tabular}{@{}lrrrr@{}}
\toprule
Statistic (\%) & 1 NE & 2 NE & 3 NE & has NE \\
\midrule
Maximin weakly dominates all NEs &
76.58 & 18.47 & 13.39 & 38.58 \\

Maximin strictly dominates all NEs &
3.86 & 0.92 & 0 & 1.83 \\

Minimax weakly dominates all NEs &
33.33 & 0 & 0 & 12 \\

Minimax strictly dominates all NEs &
0 & 0 & 0 & 0 \\

Some NE strictly dominates maximin profile &
23.42 & 70.22 & 75.69 & 54.25 \\

Maximin guarantee weakly greater than all NE guarantees &
100 & 100 & 100 & 100 \\

Maximin guarantee strictly greater than all NE guarantees &
27.28 & 13.84 & 0 & 16.46 \\

Minimax guarantee weakly greater than all NE guarantees &
46.07 & 37.62 & 35.94 & 40.39 \\

Minimax guarantee strictly greater than all NE guarantees &
12.74 & 4.29 & 0 & 6.64 \\
\bottomrule
\end{tabular}
\vspace{0.2cm}

\caption{Pareto dominance and guarantee comparisons in strictly ordinal symmetric $3\times3$ games.}
\label{tab:stats}
\end{table}

Recall that Figure~3 illustrates a game in which the maximin profile strictly Pareto dominates the Nash equilibrium while also guaranteeing a higher pay-off. More specifically, the unique pure Nash equilibrium yields $(5,5)$, while the maximin profile yields $(6,6)$. Moreover, the maximin strategy guarantees a pay-off of $4$, whereas the Nash strategy guarantees only $3$. In this example, the minimax profile coincides with the Nash equilibrium, and the corresponding Nash/minimax strategy caps the opponent's pay-off at $5$.

\end{document}